%% file: main.tex
\documentclass[a4paper,UKenglish,cleveref, autoref, thm-restate]{lipics-v2021}

\pdfoutput=1 
\hideLIPIcs  
\usepackage[linesnumbered,boxed,ruled,noend]{algorithm2e}
\usepackage{algorithmic}
\usepackage{booktabs}
\usepackage{tabularx}
\usepackage{tcolorbox}
\nolinenumbers
\usepackage{rotating} 
\newcommand{\lev}{\mathsf{lev}}
\newcommand{\plev}{\mathsf{plev}}
\newcommand{\cov}{\mathsf{cov}}
\newcommand{\asn}{\mathsf{asn}}
\theoremstyle{plain} \newtheorem{invariant}[theorem]{Invariant}

\title{Engineering Algorithms for Dynamic Greedy Set Cover} 

\author{Amitai Uzrad}{Tel Aviv University}{uzradamitai@gmail.com}{https://orcid.org/0009-0002-7519-0884}{}

\authorrunning{Amitai Uzrad}

\Copyright{Amitai Uzrad}



\supplement{}

\supplementdetails[subcategory={Source Code}, cite={}, swhid={}]{Software}{https://github.com/AmitaiUzrad/DynamicSetCover}

\funding{Funded by the European Union (ERC, DynOpt, 101043159). Views and opinions expressed are however those of the author(s) only and do not necessarily reflect those of the European Union or the European Research Council. Neither the European Union nor the granting authority can be held responsible for them. This research was also supported by the Israel Science Foundation (ISF) grant No.1991/1.}

\acknowledgements{The author is sincerely grateful to Prof. Shay Solomon for providing a supportive research environment that made this work possible.}

\begin{document}

\maketitle

\begin{abstract}
In the dynamic set cover problem, the input is a dynamic universe of elements and a fixed collection of sets. As elements are inserted or deleted, the goal is to efficiently maintain an approximate minimum set cover. While the past decade has seen significant theoretical breakthroughs for this problem, a notable gap remains between theoretical design and practical performance, as no comprehensive experimental study currently exists to validate these results.

In this paper, we bridge this gap by implementing and evaluating four greedy-based dynamic algorithms across a diverse range of real-world instances. We derive our implementations from state-of-the-art frameworks—such as [GKKP(STOC'17); SU(STOC'23); SUZ(FOCS'24)]—which we simplify by identifying and modifying intricate subroutines that optimize asymptotic bounds but hinder practical performance. We evaluate these algorithms based on solution quality (set cover size) and efficiency, which comprises \emph{update time}---the time required to update the solution following each insertion/deletion---and \emph{recourse}---the number of changes made to the solution per update. Each algorithm uses a parameter $\beta$ to balance quality against efficiency; we investigate the influence of this tradeoff parameter on each algorithm and then perform a comparative analysis to evaluate the algorithms against each other. Our results provide the first practical insights into which algorithmic strategies provide the most value in realistic scenarios.

\end{abstract}

\newpage

\input{sec1}

\input{sec2}
\input{sec3}
\input{sec4}
\input{sec5}
\input{sec6}

\bibliography{ref}

\newpage

\appendix

\input{secappendix}

\end{document}

%% file: sec1.tex
\section{Introduction}
\subparagraph{Static Set Cover.}
The minimum set cover problem is a fundamental cornerstone of combinatorial optimization. As the canonical covering problem, it encapsulates several well-known variants, including minimum vertex cover, dominating set, and edge cover. Formally, given a universe $\mathcal{U}$ of $n$ elements and a family $\mathcal{F}$ of $m$ subsets of $\mathcal{U}$, the goal is to find a minimum-sized collection $\mathcal{C} \subseteq \mathcal{F}$ such that every element in $\mathcal{U}$ is covered by $\mathcal{C}$. The \emph{frequency} $f$ of the system is the maximum number of sets that each element in $\mathcal{U}$ can be contained in. Two classic approaches define the approximation landscape for this problem: a \emph{greedy} strategy yields roughly a ($\ln n$)-approximation, while a \emph{primal-dual} approach achieves an $f$-approximation. These bounds are widely believed to be tight; unless $\text{P} = \text{NP}$, no algorithm can achieve a $(1-\epsilon) \ln n$-approximation \cite{dinur2014analytical,williamson2011design}, and assuming the Unique Games Conjecture, the $f$-approximation is likewise optimal \cite{khot2008vertex}.

\vspace{-5pt}

\subparagraph{Dynamic Set Cover.}
Over the past decade, research on the \emph{dynamic} set cover problem has expanded rapidly \cite{abboud2019dynamic,assadi2021fully,bhattacharya2025fullydynamicsetcover,bhattacharya2017deterministic,bhattacharya2015design,bhattacharya2019new,bhattacharya2021dynamic,bukov2023nearly,gupta2017online,solomon2023dynamic,solomon2025dynamicsetcoverworstcase,10756164}. In this setting, the universe $\mathcal{U}$ is modified through a sequence of update steps; in each update step, an element is either inserted into or deleted from $\mathcal{U}$. Throughout this sequence, the family of sets $\mathcal{F}$ remains static, and the universe size is bounded by $|\mathcal{U}| \leq n$. Algorithms in this field are generally categorized into the \emph{low-frequency regime} ($f < \ln n$) and the \emph{high-frequency regime} ($f \geq \ln n$). The objective is to maintain a near-optimal approximation—near $f$ or near $\ln n$, respectively—while maximizing efficiency. The most common efficiency metric is the \emph{update time}, the time required to update the solution following each insertion/deletion. Additionally, there is growing research interest in the \emph{recourse} \cite{r2,r6,r9,r15,r16,r28,r30,r32,r35,r36,gupta2017online,r38,r39,r51,ShayNoam,solomon2023dynamic}, which measures the number of changes made to the maintained solution per update step, a metric of increasing importance in dynamic environments where stability is required. Both metrics can be analyzed in terms of their \emph{amortized} (average) or \emph{worst-case} (maximum) performance.

\vspace{-5pt}

\subparagraph{Previous Results in Theory.}
 While several works have dynamized the primal-dual approach for the \emph{low-frequency} regime \cite{abboud2019dynamic,assadi2021fully,bhattacharya2025fullydynamicsetcover,bhattacharya2017deterministic,bhattacharya2015design,bhattacharya2019new,bhattacharya2021dynamic,bukov2023nearly,gupta2017online,solomon2025dynamicsetcoverworstcase,10756164}, this paper focuses on the \emph{high-frequency} regime, where algorithms dynamize the greedy approach. \Cref{comp} summarizes the key theoretical results for this regime, focusing only on those with low update time.

\begin{table}[h!]
\centering
	\caption{Summary of low update time results for dynamic set cover in the high-frequency regime. Amortized results for update time and recourse are in black, worst-case results are in blue.}\label{comp}
	\begin{tabular}{|c|c|c|c|c|c|c|}
		\hline
		Reference	&	Approx.	Factor &	Update Time	&	Recourse	&	Comments \\\hline
        \cite{gupta2017online} STOC'17	&	$c \cdot \ln n$	&	$O(f\log n)$	&	-		&	$c > 1000$	\\\hline
		\cite{solomon2023dynamic} STOC'23	&	$\beta \cdot \ln n$	&	$O\left(\frac{f \log n}{(\beta - 1)^5}\right)$	&	$O\left(\frac{1}{(\beta - 1)^4}\right)$	&	for any $\beta > 1$\\\hline
		\cite{10756164} FOCS'24	&	$\beta \cdot \ln n$	&	\textcolor{blue}{$O\left(\frac{f \log n}{(\beta - 1)^2}\right)$}	&	-	&	for any $1 < \beta < 1.25$ \\\hline
		\cite{bhattacharya2025fullydynamicsetcover} STOC'26	&	$O(\log n)$	&	\textcolor{blue}{$O(f \log^3 n)$}	&	\textcolor{blue}{$O(\log n)$}	&	- \\\hline
        \cite{solomon2025dynamicsetcoverworstcase} '25	&	$2\beta \cdot \ln n$	&	\textcolor{blue}{$O\left(\frac{f \log n}{(\beta - 1)^2}\right)$}	&	\textcolor{blue}{$O\left(\frac{1}{\beta - 1}\right)$}	&	for any $1 < \beta < 1.25$ \\\hline	
	\end{tabular}
\end{table}

\vspace{-10pt}

\subparagraph{Previous Results in Practice.} 
The field of dynamic algorithm engineering has seen substantial practical progress in recent years, with breakthrough results for several problems such as maximum independent set, edge orientation, and many more (see \cite{DBLP:conf/acda/AngrimanM0U21,DBLP:conf/acda/BorowitzG023,DBLP:conf/alenex/BorowitzG025,DBLP:conf/alenex/GoranciHLSS21,DBLP:conf/esa/GrossmannRR0HV25,DBLP:conf/alenex/GrossmannR0W25,DBLP:conf/infocom/HanauerHOS23,DBLP:conf/wea/HanauerH020,DBLP:conf/alenex/HanauerH020,DBLP:journals/jea/HanauerHS22,DBLP:conf/esa/Henzinger0P020,DBLP:conf/alenex/HenzingerN022} and references therein). However, while experimental results exist for other settings of the set cover problem—such as static and streaming (see \cite{10.1145/1871437.1871501,GAO2015750,9877844,doi:10.1137/1.9781611972795.60,WANG2021476,Belma} and references therein)—the \emph{dynamic} setting remains entirely unexplored empirically. This absence of practical study is a significant oversight, as the complex structures used to achieve optimal theoretical bounds often perform differently under real-world workloads. This gap persists for two primary reasons: first, the major theoretical breakthroughs only occurred within the last decade; second, and more importantly, these algorithms are inherently difficult to implement, often relying on intricate subroutines designed solely to preserve asymptotic guarantees rather than practical performance.

\vspace{-5pt}

\subparagraph{Our Contribution.}

This work provides the first empirical study of dynamic set cover, bridging the gap between theoretical design and practical performance. We implement and evaluate four greedy-based algorithms across several real-world instances of varying scales. Our focus is on algorithms with amortized efficiency guarantees rather than worst-case bounds. This choice is motivated by practical simplicity; algorithms with worst-case guarantees often require complex ``background'' processes that run multiple systems simultaneously (as in \cite{bhattacharya2025fullydynamicsetcover,solomon2025dynamicsetcoverworstcase,10756164}), introducing significant overhead. The algorithms presented in \Cref{comp} consist of several intricate procedures and constraints designed solely to optimize asymptotic bounds, which can hinder practical performance. We investigate which ones provide value and which constitute bottlenecks in practice, leading us to implement the following simplified versions:

\vspace{-2pt}

\begin{enumerate}
    \item \textbf{The ``robust'' algorithm}: A novel and simple algorithm based on a key property established in \cite{solomon2025dynamicsetcoverworstcase}.  
    \item \textbf{The ``local'' algorithm}: A simplified implementation of the algorithm from \cite{gupta2017online}. To achieve a practical approximation—avoiding the large ratio of $>1000 \ln n$—we sacrifice the theoretical $O(f \log n)$ update time in favor of better performance in practice.
    \item \textbf{The ``partial'' algorithm}: A simplification of the algorithm from \cite{solomon2023dynamic}. We modify complex procedures intended to reduce the theoretical update time as they often increase practical update time.  
    \item \textbf{The ``global'' algorithm}: A further simplified version of the \emph{amortized} variant in \cite{10756164}. Notably we allow $\beta \geq 1.25$, which technically weakens theoretical guarantees but is essential for significantly reducing practical update time.
\end{enumerate}

\vspace{-2pt}

\noindent Our evaluations rely on three metrics, all amortized: set cover size, update time, and recourse. Each algorithm uses a tradeoff parameter $\beta$ to balance solution quality (size) against efficiency (update time and recourse). We first evaluate the influence of $\beta$ on each algorithm individually, and then perform a comparative analysis between the four algorithms.

\vspace{-5pt}

\subparagraph{Paper Organization.}
In \Cref{staticsec}, we discuss and analyze the static greedy algorithm, its implementation, and its properties, as it serves as the foundational building block for all our dynamic algorithms. \Cref{basicsec} introduces the ``robust'' algorithm and analyzes its theoretical performance in comparison to the naive baseline which recomputes the set cover from scratch at every update step. In \Cref{levelsec}, we present and analyze the remaining three algorithms, highlighting how they differ from their original theoretical versions. \Cref{expsec} is dedicated to the experimental evaluation; we describe the instances used, detail our methodology, and present a discussion of our results. Finally, we conclude the paper in \Cref{concsec}.

%% file: sec2.tex
\section{The Static Greedy Algorithm} \label{staticsec}

In this section, we present and analyze the static greedy algorithm that serves as the foundation for all dynamic algorithms discussed in this paper. In the classic static greedy algorithm, sets are selected iteratively: in each step, we choose a set $s \in \mathcal{F}$ that covers the maximum number of uncovered elements, add it to the solution $\mathcal{C}$, and update the collection of uncovered elements. This process repeats until all elements are covered, yielding a $(\ln n + 1)$-approximation \cite{Chvatal79,Johnson74,Lovasz75}. However, as observed in \cite{10.1145/1871437.1871501}, identifying the absolute maximum set is computationally expensive in practice, requiring either a priority queue with an inverted index or multiple passes over the instance. In \cite{10.1145/1871437.1871501} it was further observed that \emph{relaxing} this greedy requirement by choosing sets that are merely \emph{close} to the maximum can drastically reduce practical running time with negligible impact on the approximation factor.

Formally, the sets are partitioned into $\lceil \log_{\beta} n \rceil + 1$ levels for a constant $\beta > 1$. Each set $s \in \mathcal{F}$ is initially placed at level $l = \lfloor \log_{\beta} |s| \rfloor$. The algorithm then iterates through levels from highest to lowest. For each set $s$ at the current level $l$, we check if it still contains at least $\beta^l$ uncovered elements. If it does, $s$ is added to $\mathcal{C}$ and its elements are marked as covered. Otherwise, $s$ is demoted to the highest level $l' < l$ such that the number of uncovered elements it contains is $\geq \beta^{l'}$; if no uncovered elements remain, the set is removed. The pseudo-code for this procedure is provided in \cref{staticalg}, called with $\mathcal{E} = \mathcal{U}$ and $\mathcal{S} = \{s \in \mathcal{F} \mid s \cap \mathcal{U} \neq \emptyset \}$.

\begin{algorithm}
	\caption{StaticGreedy($\mathcal{E}$,$\mathcal{S}$,$\beta$)} \label{staticalg}
    $\mathcal{C} \leftarrow \emptyset$\;
    \For{$s \in \mathcal{S}$}{
        Place $s$ at level $\lfloor \log_{\beta} |s \cap \mathcal{E}| \rfloor$\;
    }
	\For{$l$ from $\lceil \log_{\beta} |\mathcal{E}| \rceil$ downto $0$}{
        \While{level $l$ is not empty}{
            Choose arbitrary set $s$ from level $l$\;
            \If{$|s \cap \mathcal{E}| \geq \beta ^ l$}{
                $\mathcal{C} \leftarrow \mathcal{C} \cup \{ s \}$\;
                $\mathcal{E} \leftarrow \mathcal{E} \setminus (s \cap \mathcal{E})$\;
            }
            \Else{
                \If{$|s \cap \mathcal{E}| > 0$}{
                    Place $s$ at level $\lfloor \log_{\beta} |s \cap \mathcal{E}| \rfloor$\;
                }
                \Else{
                    Place $s$ at level $-1$\;
                }
            } 
        }
    }
    return ($\mathcal{C}$)\;    
\end{algorithm} 

\noindent The asymptotic complexity of \cref{staticalg} is $O(n + m \log_{\beta} n)$, accounting for the worst-case scenario where every set descends through all levels—a behavior rarely observed in practice. Before analyzing the approximation factor, we present the following key definitions that are central to both the static and dynamic algorithms:

\begin{definition}\label{defcov}
Denote by $\cov(s) \subseteq s$ the elements that were removed from the collection of uncovered elements due to the addition of $s$ to $\mathcal{C}$ by \cref{staticalg}. If $s \notin \mathcal{C}$ then $\cov(s) = \emptyset$. 
\end{definition}

\begin{definition}\label{deflev}
Denote by $\lev(s)$ the level from which $s$ was added to $\mathcal{C}$, where we define $\lev(s) = -1$ if $s \notin \mathcal{C}$. Similarly, for each $e \in \cov(s)$, let $\lev(e) = \lev(s)$.
\end{definition}

\begin{definition}\label{defnjs}
For any $s \in \mathcal{F}$ and $j \leq \log_{\beta} n$, $j \in \mathbb{N}$, let $N_j(s) = \{e \in s \mid \lev(e) < j\}$.
\end{definition}

\noindent Intuitively, each $\cov(s)$ contains the elements for which $s$ is ``responsible for covering''. Since every element is removed from the uncovered collection by exactly one set, the collection of all $\cov(s)$ forms a partition of $\mathcal{U}$. The following observations hold after running \cref{staticalg}:

\begin{observation}\label{obsnoND}
For any $s \in \mathcal{C}$, $|\cov(s)| \geq \beta^{\lev(s)}$.
\end{observation}

\begin{proof}
This follows immediately from Line 8 in \cref{staticalg} by definition of $\lev(s)$. 
\end{proof}

\begin{observation}\label{obsnoPD}
For any $s \in \mathcal{F}$ and $j \leq \log_{\beta} n$, $j \in \mathbb{N}$ we have $|N_j(s)| < \beta^{j}$.
\end{observation}

\begin{proof}
Assume by contradiction that exist $s$ and $j$ such that $|N_j(s)| \geq \beta^{j}$. If $\lev(s) \geq j$, then when $s$ was added to $\mathcal{C}$, every uncovered element $e \in s$ was assigned to $\cov(s)$ and thus obeys $\lev(e) \geq j$; so $N_j(s) = \emptyset$, a contradiction. If $\lev(s) < j$, then right after level $j$ was processed, $s$ was not in $\mathcal{C}$ and was at some level $j' < j$. At that time, all elements in $N_j(s)$ were uncovered by definition, so $|s \cap \mathcal{E}| \geq |N_j(s)| \geq \beta^j$. However, whenever $s$ is placed at a new level (lines 4 and 13), we assign it to the \emph{highest} level $l'$ such that the number of uncovered elements it contains is $\geq \beta^{l'}$. Since the number of uncovered elements never increases, we cannot have $|s \cap \mathcal{E}| \geq \beta^j$ while $s$ remains at a level less than $j$, a contradiction.
\end{proof}

\noindent Informally, \cref{obsnoND} shows that every set is positioned at its ``correct'' level, while \cref{obsnoPD} implies that no set can ascend to a higher level by claiming low-level elements and increasing their levels. Essentially, all elements are positioned ``as high as possible'', which captures the greedy principle in action: elements are effectively chosen at the earliest possible opportunity. As it turns out, Observations \ref{obsnoND} and \ref{obsnoPD} suffice for the approximation analysis, and this point will be emphasized after the following approximation factor lemma.

\begin{lemma} \label{lemapprox}
\cref{staticalg} produces a solution $\mathcal{C}$ with approximation factor $\beta \cdot (\ln (n) + 1)$.
\end{lemma}

\begin{proof}
For each element $e$ define $q_e = \beta^{-\lev(e)}$. For each set $s \in \mathcal{C}$, we know by \cref{obsnoND} that $|\cov(s)| \geq \beta^{\lev(s)}$, and thus $1 \leq \sum_{e \in \cov(s)} q_e$. Summing up on all sets in $\mathcal{C}$ we obtain:
\begin{equation} \label{eq1}
    \hfill |\mathcal{C}| \leq \sum_{e \in \mathcal{U}} q_e. \hfill
\end{equation}
\noindent Denote by $s_1, s_2, \ldots, s_k$ the sets in an optimal set cover, and consider one of these sets $s_i$. Let $s_i = \{e_1, e_2, \ldots, e_{x_i} \}$, where $x_i = |s_i|$. Notice that each element $e_j \in s_i$ is not necessarily in $\cov(s_i)$. Without loss of generality assume that $\lev(e_1) \geq \lev(e_2) \geq \ldots \geq \lev (e_{x_i})$.  

\begin{claim} \label{claimapprox}
For any $e_j \in s_i$ we have that $q_{e_j} \leq \frac{\beta}{x_i - j + 1}$.
\end{claim}
\begin{claimproof}
Assume by contradiction that $q_{e_j} > \frac{\beta}{x_i - j + 1}$ and denote $l = \lev(e_j)$. Rearranging, we get that $x_i - j + 1 > \beta^{l+1}$. Since $\lev(e_{j'}) \leq l$ for any $j' \geq j$, we know that all the $(x_i - j + 1)$ elements $e_j,e_{j+1}, \ldots, e_{x_i}$ are at a level $\leq l$. Since all these elements are in $s_i$, we know that $|N_{l+1}(s_i)| \geq x_i - j + 1 > \beta^{l+1}$, contradicting \cref{obsnoPD}.
\end{claimproof}

\noindent Equipped with \cref{claimapprox}, summing up on all $e \in s_i$ we get that $\sum_{e \in s_i} q_e \leq \beta \cdot (1 + \frac{1}{2} + \frac{1}{3} + \ldots + \frac{1}{x_i}) \leq \beta \cdot H_n \leq \beta \cdot (\ln (n) + 1)$, where $H_n$ is the $n$'th harmonic number. By summing up on all $k$ sets in an optimal solution, we obtain:
\begin{equation} \label{eq2}
    \hfill \sum_{i=1}^k \sum_{e \in s_i} q_e \leq \sum_{i=1}^k \beta \cdot (\ln (n) + 1) = k \cdot \beta \cdot (\ln (n) + 1). \hfill
\end{equation}
\noindent Since an optimal solution covers all elements, we know that:
\begin{equation} \label{eq3}
    \hfill \sum_{e \in \mathcal{U}} q_e \leq \sum_{i=1}^k \sum_{e \in s_i} q_e. \hfill
\end{equation}
\noindent Combining Equations \ref{eq1}, \ref{eq2} and \ref{eq3}, we get that $\frac{|\mathcal{C}|}{k} \leq \beta \cdot (\ln (n) + 1)$, concluding the proof.
\end{proof}

\noindent Thus, the approximation factor incurs only a slight overhead of $\beta$ compared to the classic greedy algorithm. The \textbf{key insight} is that this result relies \textbf{solely} on Observations \ref{obsnoND} and \ref{obsnoPD}. Consequently, \textbf{any} partition of $\mathcal{U}$ into $\cov$ sets and corresponding levels that satisfies these two observations will yield a $\beta \cdot (\ln n + 1)$ approximation by including all sets with non-empty $\cov$ sets, \textbf{regardless of how it was obtained}. This is crucial because the dynamic algorithms presented in \cref{levelsec} \emph{maintain} an element partition that satisfies a ``relaxed'' version of these observations. We conclude this section with the following theorem:

\begin{nolinenumbers}
\begin{tcolorbox} [width=\linewidth, sharp corners=all, colback=white!95!black]
\begin{theorem} \label{thm}
Consider an assignment of each element $e \in \mathcal{U}$ to a set $s \in \mathcal{F}$ containing it, and for each $s \in \mathcal{F}$ denote by $\cov(s)$ the elements assigned to $s$. Place each set $s$ such that $\cov(s) \neq \emptyset$ at an integer level $\lev(s) \in [0,\lceil \log_{\beta} n \rceil]$ (for constant $\beta > 1$), and place each $e \in \cov(s)$ at level $\lev(e) = \lev(s)$. $\mathcal{C} = \{s \in \mathcal{F} \mid \cov(s) \neq \emptyset\}$ is a $\beta \cdot (\ln n + 1)$-approximate set cover if the following two properties hold:

\vspace{-3pt}

\begin{enumerate}
    \item For any $s \in \mathcal{F}$ such that $\cov(s) \neq \emptyset$ we have $|\cov(s)| \geq \beta^{\lev(s)}$.
    \item For any $s \in \mathcal{F}$ and $j \leq \log_{\beta} n$, $j \in \mathbb{N}$ we have $|\{e \in s \mid \lev(e) < j\}| < \beta^{j}$.
\end{enumerate}
\end{theorem}
\end{tcolorbox}
\end{nolinenumbers}

%% file: sec3.tex
\section{Basic Dynamic Algorithms} \label{basicsec}
In the dynamic setting, the collection of sets $\mathcal{F}$ is fixed, but elements are inserted into or deleted from $\mathcal{U}$ at each update step. The \emph{naive baseline} re-runs \cref{staticalg} after every update; while this maintains a $\beta \cdot (\ln n + 1)$ approximation, this is extremely inefficient. Instead, our first dynamic algorithm leverages a ``robustness'' property established by \cite{solomon2025dynamicsetcoverworstcase}. Informally, we execute the static greedy algorithm only intermittently and simply maintain a valid solution between runs, significantly reducing the update time.

\subsection{The ``Robust'' Algorithm}
We begin with a key claim that forms the basis of our first dynamic algorithm.
\begin{claim} \label{claimrobust}
Let $\mathcal{C}$ be a solution produced by \cref{staticalg}. For any $1 < \beta < 2$ and any collection of deleted elements $\mathcal{D} \subset \mathcal{U}$ such that $|\mathcal{D}| \leq (\beta - 1) \cdot |\mathcal{C}|$, the optimum set cover size $k'$ for the remaining elements $\mathcal{U} \setminus \mathcal{D}$ obeys $k' \geq \frac{(2-\beta) \cdot |\mathcal{C}|}{\beta \cdot (\ln n + 1)}$.
\end{claim}
\begin{claimproof}
Denote by $s_1, s_2, \ldots, s_{k'}$ the sets in an optimal set cover for the new system with universe $\mathcal{U} \setminus \mathcal{D}$. Since \cref{obsnoPD} remains valid, the logic from \cref{lemapprox} still yields:
\vspace{-5pt}
\begin{equation} \label{eq4}
\hfill \sum_{e \in \mathcal{U} \setminus \mathcal{D}} q_e \leq \sum_{i=1}^{k'} \sum_{e \in s_i} q_e \leq k' \cdot \beta \cdot (\ln n + 1). \hfill
\end{equation}
\noindent While \cref{obsnoND} may no longer hold following the deletions, we have by \Cref{eq1}:
\begin{equation} \label{eq5}
    \hfill |\mathcal{C}| \leq \sum_{e \in \mathcal{U}} q_e = \sum_{e \in \mathcal{U} \setminus \mathcal{D}} q_e + \sum_{e \in \mathcal{D}} q_e \leq \sum_{e \in \mathcal{U} \setminus \mathcal{D}} q_e + |\mathcal{D}| \leq \sum_{e \in \mathcal{U} \setminus \mathcal{D}} q_e + (\beta - 1) \cdot |\mathcal{C}|. \hfill
\end{equation}
Rearranging and substituting \cref{eq4} gives the desired ratio, concluding the claim.
\end{claimproof}

\noindent The ``robust'' algorithm operates in discrete intervals. At the start of the $i$-th interval, we compute a set cover $\mathcal{C}_i$ and set the interval length to $(\beta - 1) \cdot |\mathcal{C}_i|$ update steps. Throughout this period, the algorithm naively maintains feasibility: if an inserted element $e$ is not covered by the current solution, an arbitrary set $s \ni e$ is added to the cover. Once the interval ends, we re-run the static greedy algorithm to obtain $\mathcal{C}_{i+1}$ and initiate the next interval. \Cref{robustalg} provides the pseudo-code for a single update step, where $e$ is the updated element, $\mathcal{S}(e)$ is the collection of sets containing it and $r$ is a parameter used to track the remaining steps in the current interval. The algorithm uses \Cref{staticalg} as a subroutine.

\begin{algorithm}
	\caption{UpdateRobust($e$,$\mathcal{S}(e)$)} \label{robustalg}
    \If{insertion}{
        $\mathcal{U} \leftarrow \mathcal{U} \cup \{e\}$\;
        \If{$\mathcal{C} \cap \mathcal{S}(e) = \emptyset$}{
            Choose arbitrary $s \in \mathcal{S}(e)$\;
            $\mathcal{C} \leftarrow \mathcal{C} \cup \{s\}$\;
        }
    }
    \Else{
        $\mathcal{U} \leftarrow \mathcal{U} \setminus \{e\}$\;
    }
    $r \leftarrow r - 1$\;
    \If{$r = 0$}{
        $\mathcal{S} \leftarrow \{s \in \mathcal{F} \mid s \cap \mathcal{U} \neq \emptyset \}$\;
        $r \leftarrow (\beta - 1) \cdot$|StaticGreedy($\mathcal{U}$,$\mathcal{S}$,$\beta$)|\;
    }
\end{algorithm}

\vspace{-10pt}

\begin{claim} \label{claimrobustalg}
\cref{robustalg} maintains a $\left(\frac{\beta^2}{2 - \beta} \cdot (\ln n + 1)\right)$-approximate set cover $\mathcal{C}$.
\end{claim}
\begin{claimproof}
Each insertion can increase the solution size by at most one and never decreases the optimum size. During the $i$-th interval, which lasts $(\beta - 1) \cdot |\mathcal{C}_i|$ steps, the solution size remains $\leq \beta \cdot |\mathcal{C}_i|$. Applying the bound from \cref{claimrobust} yields the claimed approximation.
\end{claimproof}

\noindent The amortized update time and recourse of the ``robust'' algorithm depend primarily on the average solution size. Large solutions result in fewer re-runs of the static greedy algorithm, minimizing both amortized update time and recourse. Conversely, small solutions shorten the intervals, diminishing the efficiency gains over the naive baseline.

%% file: sec4.tex
\section{Level-Based Dynamic Algorithms} \label{levelsec}

In this section we present three dynamic algorithms that are simplifications of known theoretical algorithms. Although they are far more complex than the ``robust'' algorithm, they are based on maintaining an assignment of elements to sets, as well as elements and sets to levels, that satisfies a \emph{relaxation} of the two properties defined in \Cref{thm}. The distinction between them lies in how they implement these relaxations. The first algorithm (\Cref{localsec}) maintains a \emph{local} relaxation of both properties in \Cref{thm}. The second algorithm (\Cref{partialsec}) maintains a \emph{global} relaxation of the first and a \emph{local} relaxation of the second. The third algorithm (\Cref{globalsec}) maintains a \emph{global} relaxation of both. Intuitively, a local relaxation signifies that each individual set satisfies the property within some added slack. In contrast, a global relaxation implies that the property is maintained ``on average'' across the entire collection of sets. 
In all three algorithms, the maintained set cover $\mathcal{C}$ consists of all sets $s \in \mathcal{F}$ such that $\cov(s) \neq \emptyset$. This construction \emph{ensures that a valid solution is always maintained}, even in cases where our simplifications modify the procedures required to preserve certain theoretical asymptotic guarantees.

\subsection{The ``Local'' Algorithm} \label{localsec}
In this section we introduce the ``local'' algorithm, a simplification of \cite{gupta2017online}, which maintains a per-set relaxation of the two properties in \Cref{thm}. We begin with the following definitions:

\begin{definition} \label{localdef}
A set $s$ is denoted ``negative dirty'' (ND) if $|\cov(s)| < \beta^{\lev(s) - 1}$, and denoted ``positive dirty with respect to $j$'' ($j$-PD) if $|N_j(s)| \geq \beta^{j + 1}$. In addition, for each element $e$ let $\asn(e)$ be the set $s$ such that $e \in \cov(s)$.
\end{definition}

\noindent Upon insertion of element $e$, it is assigned to the set containing it $s$ with the highest level $\lev(s)$. The local algorithm maintains the following invariant:

\begin{invariant} \label{invlocal}
No set is ND, and no set is $j$-PD for any level $j$.
\end{invariant}

\noindent If \Cref{invlocal} holds, then the properties of \Cref{thm} are satisfied within a factor of $\beta$. Specifically, every $s \in \mathcal{C}$ satisfies $|\cov(s)| \geq \beta^{\lev(s) - 1}$, and every $s \in \mathcal{F}$ satisfies $|N_j(s)| < \beta^{j + 1}$ for all levels $j$. Applying the logic from \Cref{lemapprox} with the extra slack yields the following:

\begin{lemma} \label{lemlocal}
Any solution satisfying \Cref{invlocal} has approximation factor $\beta^3 \cdot (\ln n + 1)$.
\end{lemma}

\vspace{-10pt}

\begin{algorithm}
	\caption{UpdateLocal($e$,$\mathcal{S}(e)$)} \label{localalg}
    \If{insertion}{
        $\mathcal{U} \leftarrow \mathcal{U} \cup \{e\}$\;
        $s \leftarrow \mathsf{argmax} \{\lev(s') : s' \in \mathcal{S}(e)\}$ (breaking ties arbitrarily)\;
        $\mathcal{C} \leftarrow \mathcal{C} \cup \{s\}$; \: $\cov(s) \leftarrow \cov(s) \cup \{e\}$\;
        RisingPhase($\mathcal{S}(e)$)\; 
    }
    \Else{
        $\mathcal{U} \leftarrow \mathcal{U} \setminus \{e\}$\;
        $s \leftarrow \asn(e)$\;
        $\cov(s) \leftarrow \cov(s) \setminus \{e\}$\;
        FallingPhase($\{s\}$)\;
    }
\end{algorithm}

\begin{nolinenumbers}
\noindent 
\begin{minipage}{0.47\textwidth}
\begin{algorithm}[H]
	\caption{RisingPhase($\mathcal{S}$)} \label{risingalg}
    $\mathcal{S}' \leftarrow \emptyset$\;
    \For{$j$ from $\lceil \log_{\beta} n \rceil$ downto $0$}{
        \For{$s \in \mathcal{S}$}{
            \If{$s$ is $j$-PD}{
                Rise $s$ to level $j+1$\;
                $\cov(s) \leftarrow N_j(s)$\;
                \For{$e'$ in $\cov(s)$}{
                    $s' \leftarrow$ previous $\asn(e')$\;
                    $\mathcal{S}' \leftarrow \mathcal{S}' \cup \{s'\}$\;
                }
            }
        }
    }
    \If{$\mathcal{S}' \neq \emptyset$}{
        FallingPhase($\mathcal{S}'$)\;
    }
\end{algorithm}
\end{minipage}
\hfill
\begin{minipage}{0.47\textwidth}
\begin{algorithm}[H]
	\caption{FallingPhase($\mathcal{S}'$)} \label{fallingalg}
    $\mathcal{S} \leftarrow \emptyset$\;
    \For{$s' \in \mathcal{S}'$}{
        \If{$s'$ is ND}{
            $j = \lfloor \log_{\beta} |\cov(s')| \rfloor$\;
            Drop $s'$ to level $j$\;
            Recreate $\cov(s')$\;
            \For{$e' \in \cov(s')$}{
                \For{$s \ni e'$}{
                    $\mathcal{S} \leftarrow \mathcal{S} \cup \{s\}$\;
                }
            }
        }
    }
    \If{$\mathcal{S} \neq \emptyset$}{
        RisingPhase($\mathcal{S}$)\;
    }
\end{algorithm}
\end{minipage}
\end{nolinenumbers}

\vspace{10pt}

\noindent The local algorithm maintains \Cref{invlocal} through two mutually recursive procedures, the \emph{rising phase} and the \emph{falling phase}. In the rising phase, we iterate through the sets in $\mathcal{S}$, an input collection of sets, and identify any that are $j$-PD, processing from the highest level $j$ down to the lowest. When a $j$-PD set $s$ is found, we ``clean'' it by raising its level to $j+1$ and reassigning its elements such that $\cov(s) = N_j(s)$. After processing all levels, there are no more $j$-PD sets (for any $j$), and we identify every element $e'$ whose level changed. Because these elements may have left their original $\cov$ sets, those sets may now be ND. We collect these potentially ND sets into a collection $\mathcal{S}'$ and initiate a falling phase. In the falling phase, we examine the sets in $\mathcal{S}'$. Any ND set $s'$ is dropped to level $\lfloor \log_{\beta} |\cov(s')| \rfloor$, thus there are no ND sets at the end of this process. This level shift causes elements $e'$ to drop to a lower level, which in turn might make any set containing $e'$ become $j$-PD. Consequently, we collect all sets containing these elements into a collection $\mathcal{S}$ and trigger a new rising phase. This process continues until termination, where eventually our system obeys \Cref{invlocal}. 

An insertion of $e$ triggers an initial rising phase with $\mathcal{S} = \mathcal{S}(e)$ (the sets containing $e$), while a deletion of $e$ triggers an initial falling phase with $\mathcal{S}' = \{\asn(e)\}$. The logic for these procedures is detailed in Algorithms \ref{localalg}, \ref{risingalg}, and \ref{fallingalg}. While \cite{gupta2017online} demonstrates that these ``clean-up'' procedures terminate with efficient amortized update times for very large values of $\beta$, such values are impractical for most real-world instances. Since we employ significantly lower values—typically $\beta < 2$—we provide a proof that this process is guaranteed to terminate.

\begin{claim} \label{terminate}
\Cref{localalg} terminates after every update step.
\end{claim}

\begin{proof}
Consider a rising phase where $j$ is the maximum level such that some set $s \in \mathcal{S}$ is $j$-PD. Since the highest possible rise is to level $j+1$, only elements at levels $< j$ can be reassigned. Consequently, any set entering $\mathcal{S}'$ for the subsequent falling phase must be at a level strictly less than $j$. During that falling phase, the count $|N_{j'}(s')|$ for any $s' \in \mathcal{S}'$ can only increase for levels $j' < j$. It follows that in the next rising phase, the maximum level $j'$ for which any set can be $j'$-PD is at most $j-1$. Thus, the highest rising level decreases with each mutual recursion, eventually terminating.
\end{proof}

\subsection{The ``Partial'' Algorithm} \label{partialsec}

In this section we introduce the ``partial'' algorithm, a simplification of \cite{solomon2023dynamic} that maintains a \emph{global} relaxation of the first property alongside a \emph{local} relaxation of the second property in \Cref{thm}. In addition to the terms in \Cref{localdef}, we define the following:

\begin{definition} \label{partialdef}
For each level $j$, let $\mathcal{D}_j$ be a ``dirt counter''. Let $\mathcal{D} = \sum_j \mathcal{D}_j$ be the total dirt across all levels. Whenever an element leaves some $\cov(s)$, we increment $\mathcal{D}_{\lev(s)}$ by $\beta^{-\lev(s)}$.
\end{definition}

\noindent Upon insertion of element $e$ where $\mathcal{S}(e)$ is the collection of sets containing it, it is assigned to the set $s \in \mathcal{S}(e)$ with the highest level $\lev(s)$. The partial algorithm maintains the following invariant:

\begin{invariant} \label{invpartial}
No set is $j$-PD for any level $j$, and $\mathcal{D} < \frac{\beta - 1}{\beta} \cdot |\mathcal{C}|$.
\end{invariant}

\noindent The dirt counters serve as a global relaxation of the ``no ND sets'' invariant. A low total dirt $\mathcal{D}$ implies that sets are not ND ``on average'' in a sense. However, the system accumulates dirt in two ways: when an element is deleted, and when a set becomes $j$-PD and rises, forcing elements to leave their previous $\cov$ sets. Once the total dirt exceeds the threshold $\mathcal{D} \geq \frac{\beta - 1}{\beta} \cdot |\mathcal{C}|$, the algorithm triggers a partial rebuild. Instead of a full reset, we identify a \emph{critical level} $i_{\mathsf{crit}}$ and run the static greedy algorithm only on the subsystem comprising elements and sets at or below level $i_{\mathsf{crit}}$. This process ``cleans'' the subsystem, allowing us to reset all dirt counters up to $i_{\mathsf{crit}}$ to zero. We will not define here what the critical level exactly is or how to find it, but intuitively it is chosen to maximize the ``dirt-to-cost'' ratio. It essentially is the highest level where the accumulated dirt is disproportionately large relative to the subsystem's size, thus maximizing the progress made toward a ``clean'' system while minimizing the computational overhead of the static greedy re-run. See \cite{solomon2023dynamic} for the formal definition of the critical level, along with the proof for the following lemma:

\begin{lemma} [Lemma 39 in \cite{solomon2023dynamic}] \label{lempartial}
Any solution satisfying \Cref{invpartial} has approximation factor $\beta^{O(1)} \cdot (\ln n + 1)$.
\end{lemma}

\noindent While \cite{solomon2023dynamic} achieves an amortized update time of $O \left( \frac{f \cdot \log n}{\text{poly}(\beta - 1)}\right)$, this bound requires several complex procedures—specifically for locating the critical level and maintaining $N_j(s)$ for all $j$ and $s$. Although these procedures are vital for the theoretical $\log n$ factor, they significantly overcomplicate the algorithm's practical implementation. Consequently, we employ a simplified version that achieves a slightly higher amortized update time of $O \left( \frac{f \cdot \log^2 n}{\text{poly}(\beta - 1)}\right)$. Notably, our implementation retains the amortized recourse of $O \left( \frac{1}{\text{poly}(\beta - 1)}\right)$ established in \cite{solomon2023dynamic}. \Cref{partialalg} provides the pseudo-code for a single update step of this algorithm.

\vspace{10pt}

\begin{nolinenumbers}
\noindent 
\begin{minipage}{0.48\textwidth}
\begin{algorithm}[H]
	\caption{UpdatePartial($e$,$\mathcal{S}(e)$)} \label{partialalg}
    \If{insertion}{
        $\mathcal{U} \leftarrow \mathcal{U} \cup \{e\}$\;
        $s = \mathsf{argmax} \{\lev(s') : s' \in \mathcal{S}(e)\}$\;
        $\mathcal{C} \leftarrow \mathcal{C} \cup \{s\}$\;
        $\cov(s) \leftarrow \cov(s) \cup \{e\}$\;
        \If{exist PD sets $s' \in \mathcal{S}(e)$}{
            Rise highest level PD set\;
            \For{each rising element $e'$}{
                $l' \leftarrow$ previous $\lev(e')$\;
                $\mathcal{D}_{l'} \leftarrow \mathcal{D}_{l'} + \beta^{-l'}$\;
            }  
        } 
    }
    \Else{
        $\mathcal{U} \leftarrow \mathcal{U} \setminus \{e\}$\;
        $s = \asn(e)$\;
        $\cov(s) \leftarrow \cov(s) \setminus \{e\}$\;
        $\mathcal{D}_{\lev(s)} \leftarrow \mathcal{D}_{\lev(s)} + \beta^{-\lev(s)}$\;
    }
    \If{$\mathcal{D} \geq \frac{\beta - 1}{\beta} \cdot |\mathcal{C}|$}{
        Find critical level $i_{\mathsf{crit}}$\;
        $\mathcal{E} \leftarrow \{e' \mid \lev(e') \leq i_{\mathsf{crit}}\}$\;
        $\mathcal{S} \leftarrow \{s' \in \mathcal{F} \mid s' \cap \mathcal{E} \neq \emptyset\}$\;
        $\mathcal{C}' \leftarrow$ StaticGreedy($\mathcal{E}$,$\mathcal{S}$,$\beta$)\;
        $\mathcal{D}_i \leftarrow 0$, $\forall i \leq i_{\mathsf{crit}}$\;
    }
\end{algorithm}
\end{minipage}
\hfill
\begin{minipage}{0.48\textwidth}
\begin{algorithm}[H]
	\caption{UpdateGlobal($e$,$\mathcal{S}(e)$)} \label{globalalg}
    \If{insertion}{
        $\mathcal{U} \leftarrow \mathcal{U} \cup \{e\}$\;
        $s = \mathsf{argmax} \{\lev(s') : s' \in \mathcal{S}(e)\}$\;
        $\mathcal{C} \leftarrow \mathcal{C} \cup \{s\}$\;
        $\cov(s) \leftarrow \cov(s) \cup \{e\}$\;
        \For{$i$ from $\lev(e)$ to $\lceil \log_{\beta} n \rceil$}{
            $\mathcal{P}_i \leftarrow \mathcal{P}_i + 1$\;
        }
    }
    \Else{
        \For{$i$ from $\lev(e)$ to $\plev(e) - 1$}{
            $\mathcal{A}_i \leftarrow \mathcal{A}_i - 1$; \: $\mathcal{D}_i \leftarrow \mathcal{D}_i + 1$\;
        }
        \For{$i$ from $\plev(e)$ to $\lceil \log_{\beta} n \rceil$}{
            $\mathcal{P}_i \leftarrow \mathcal{P}_i - 1$; \: $\mathcal{D}_i \leftarrow \mathcal{D}_i + 1$\;
        }
        $\mathcal{U} \leftarrow \mathcal{U} \setminus \{e\}$\;
        $s = \asn(e)$\;
        $\cov(s) \leftarrow \cov(s) \setminus \{e\}$\;
    }
    \If{$\exists i \mid \mathcal{P}_i + \mathcal{D}_i > 2(\beta - 1) \cdot \mathcal{A}_i$}{
        Find highest $i$, denote by $i_{\mathsf{crit}}$\;
        $\mathcal{E} \leftarrow \{e' \mid \lev(e') \leq i_{\mathsf{crit}}\}$\;
        $\mathcal{S} \leftarrow \{s' \in \mathcal{F} \mid s' \cap \mathcal{E} \neq \emptyset\}$\;
        $\mathcal{C}' \leftarrow$ StaticGreedy($\mathcal{E}$,$\mathcal{S}$,$\beta$)\;
        Update $\mathcal{A}_i$, $\mathcal{P}_i$, $\mathcal{D}_i$, $\forall i \leq i_{\mathsf{crit}}$\;
    }
\end{algorithm}
\end{minipage}
\end{nolinenumbers}

\subsection{The ``Global'' Algorithm} \label{globalsec}

In this section we introduce the ``global'' algorithm, a simplification of \cite{10756164} that maintains a \emph{global} relaxation of both properties in \Cref{thm}. We define the following:

\begin{definition} \label{globaldef}
Each element $e$ is assigned a ``passive level'', $\plev(e)$, such that $\plev(e) \geq \lev(e)$ at all times. Over the lifespan of $e$, $\plev(e)$ is monotonically non-decreasing. For each level $i$, let $\mathcal{A}_i = |\{e \in \mathcal{U} \mid \lev(e) \leq i < \plev(e)\}|$, $\mathcal{P}_i = |\{e \in \mathcal{U} \mid \plev(e) \leq i\}|$ and let $\mathcal{D}_i$ be a dirt counter that counts the number of deleted elements $e$ such that $\lev(e) \leq i$ when deleted.
\end{definition}

\noindent Upon insertion of element $e$ where $\mathcal{S}(e)$ is the collection of sets containing it, it is assigned to the set $s \in \mathcal{S}(e)$ with the highest level $\lev(s)$,
and we set $\plev(e) = \lev(e)$, which raises $\mathcal{P}_i$ for all $i \geq \lev(e)$. When $e$ is deleted, we update all $\mathcal{A}_j$, $\mathcal{P}_j$ and $\mathcal{D}_j$ according to their respective definitions. Intuitively, for any level $j$, $\mathcal{P}_j$ and $\mathcal{D}_j$ track ``dirt'' from insertions and deletions, respectively, while $\mathcal{A}_j$ serves as a ``clean'' counter. The algorithm maintains the following global invariant:

\begin{invariant} \label{invglobal}
For any level $i$, $\mathcal{P}_i + \mathcal{D}_i \leq 2(\beta - 1) \cdot \mathcal{A}_i$.
\end{invariant}

\noindent Once \Cref{invglobal} is violated, the algorithm triggers a partial rebuild, exactly like in the partial algorithm, up to a \emph{critical level} $i_{\mathsf{crit}}$, which here is the highest violating level. This process ``cleans'' the subsystem, setting the passive level of each element $e$ participating in the rebuild to be $\max \{\plev(e) , i_{\mathsf{crit}} + 1\}$,
causing elements to leave $\mathcal{P}_i$ and join $\mathcal{A}_i$, and sets $\mathcal{D}_i = 0$ for all $i \leq i_{\mathsf{crit}}$. 
Since we implement this algorithm with $\beta \geq 1.25$, the approximation guarantees from \cite{10756164} are no longer valid. In addition, as with the partial algorithm, we prioritize a practical implementation over the more complex version in \cite{10756164} that targets optimal asymptotic bounds. Consequently, our algorithm achieves an amortized update time of $O \left( \frac{f \cdot \log^2 n}{\text{poly}(\beta - 1)}\right)$. \Cref{globalalg} details the pseudo-code for a single update step. See Section 2.2.1 in \cite{10756164} for more details regarding the theoretical amortized version.

\bigskip

\noindent We summarize \Cref{levelsec} in \Cref{comp2}. 

\begin{table}[h!]
\centering
	\caption{Summary of the practical variants of the level-based algorithms, with their guarantees.}\label{comp2}
	\begin{tabular}{|c|c|c|c|c|c|c|}
		\hline
		Ref.	&	Practical Variant &	Approx. Factor	&	Update Time	&  Recourse \\\hline
        \cite{gupta2017online}	&	Local Algorithm	&	$\beta^3 \cdot (\ln n + 1)$	&	-	& -	\\\hline
		\cite{solomon2023dynamic}	&	Partial Algorithm	&	-	&	$O \left( \frac{f \cdot \log^2 n}{\text{poly}(\beta - 1)}\right)$ & $O \left( \frac{1}{\text{poly}(\beta - 1)}\right)$	\\\hline
		\cite{10756164}	&	Global Algorithm	&	-	&	$O \left( \frac{f \cdot \log^2 n}{\text{poly}(\beta - 1)}\right)$	& - \\\hline
	\end{tabular}
\end{table}

%% file: sec5.tex
\section{Experimental Evaluation} \label{expsec}

\subparagraph{Instances.}
We evaluated our four algorithms on 120 real-world hypergraph instances. We treat the dynamic set cover problem as a dynamic vertex cover problem on hypergraphs, where dynamic hyperedges represent elements and static vertices represent sets. We started out with hundreds of static hypergraph instances \cite{Benson-2018-subset,10.1145/2049662.2049663,ni2019justifying}. To ``dynamize'' a static instance with $x$ elements, we generated a dynamic sequence of length $k=2x$, ensuring each element is inserted and deleted exactly once (staring and ending with an empty system). Insertions follow the original static ordering. We enforce a ``capacity constraint'' of $n = \lfloor x/10 \rfloor$ active elements. While the system is below capacity, we perform an insertion with 80\% probability and a deletion of one of the five most recently inserted elements with 20\% probability. Once capacity $n$ is reached, we trigger a ``cleanup'' phase, deleting the $d$ oldest active elements, where $d \sim \text{Uniform}(1, \lfloor n/10 \rfloor)$. This model specifically evaluates the performance against a mix of ``fresh'' deletions, ``stale'' deletions, and varying batch sizes. Finally, we filtered the collection to include only high-frequency instances ($f \geq \ln n$), resulting in the 120 instances used for benchmarking. See Tables \ref{table1}, \ref{table2} and \ref{table3} for basic statistics of these instances.

\vspace{-5pt}

\subparagraph{Methodology.}
We implemented the algorithms using C++ and compiled them using gcc 11 with full optimization (-O3 flag). The experimental evaluation was conducted on an AMD EPYC 7763 processor (64 cores, 2.45GHz) with 512MB of L3 cache and 256GB of main memory. We executed up to 32 experiments in parallel. 

To reflect the ``amortized nature'' of the algorithms, we used three quality metrics, all amortized: set cover size, update time (processing and execution), and recourse. We measured these values for every update step of each individual experiment. The final quality metrics per experiment were calculated by averaging these values over the entire update sequence.
For each of the four algorithms we chose 17 different values for $\beta$, to test the tradeoff between our quality metrics. For the smallest 115 instances (of the 120) we ran each algorithm and $\beta$ combination five times (taking an average), totaling to $115 \cdot 4 \cdot 17 \cdot 5 = 39,100$ individual experiments. Due to the longer running time of the last 5 instances, we reduced the number of runs per combination and the number of different $\beta$ values per algorithm, and set a time limit of 48 hours for each individual run. Overall there were roughly $40,000$ experiments run. 

We compare our algorithms using \emph{performance profiles} \cite{dolan2004benchmarkingoptimizationsoftwareperformance}. Given a quality metric and a collection of algorithms, this plot relates for each algorithm the \% of instances where it is within a $\tau$ factor of the best algorithm; the y-axis represents the \% of instances and the x-axis represents $\tau$. For example, if the ``partial'' algorithm reaches 70\% at $\tau = 2.5$ in the amortized update time profile, this means that in 70\% of the instances, the update time of the ``partial'' algorithm is within a $2.5$ factor of the best (in this case fastest) algorithm.

\vspace{-5pt}

\subparagraph{Results.}
Since we tested 17 different values of $\beta$ for each of the four algorithms, a single visualization of all 68 variants would be cluttered. We therefore adopted a two-step analysis. In the first step, we selected six $\beta$ values per algorithm to conduct an intra-algorithm comparison of the impact of $\beta$. The performance profiles for each of these algorithms are displayed in \Cref{figpp4} (in \Cref{secapp}), and they confirm the expected correlations: $\beta$ is strongly positively correlated with the amortized set cover size and strongly negatively correlated with amortized update time and recourse across all algorithms. Furthermore, the results highlight the relative sensitivity of each metric to $\beta$; the impact is most pronounced for update time (high $\tau$ values), followed by recourse, and finally set cover size (low $\tau$ values).

Before proceeding to the second step, and to further investigate the trade-off between solution quality and computational efficiency, we provide a complementary analysis in \Cref{fig6}. Since the raw metrics vary significantly across our 120 instances, we normalize the results per instance to ensure an unbiased comparison. Specifically, for each instance and each metric separately, we identify the best-performing algorithm and $\beta$ combination. We then calculate the performance of all other combinations relative to these best-found values for each instance. Finally, we aggregate these normalized ratios across all instances using the geometric mean. \Cref{fig6} displays the resulting quality vs. efficiency curves for the four algorithms across the full range of $\beta$ values. These curves explicitly visualize the negative correlation between solution quality and computational efficiency. Furthermore, consistent with \Cref{figpp4}, they highlight the relative sensitivity of each metric to the parameter $\beta$.

\begin{figure}[h!]
	\center{\includegraphics[scale=0.32]{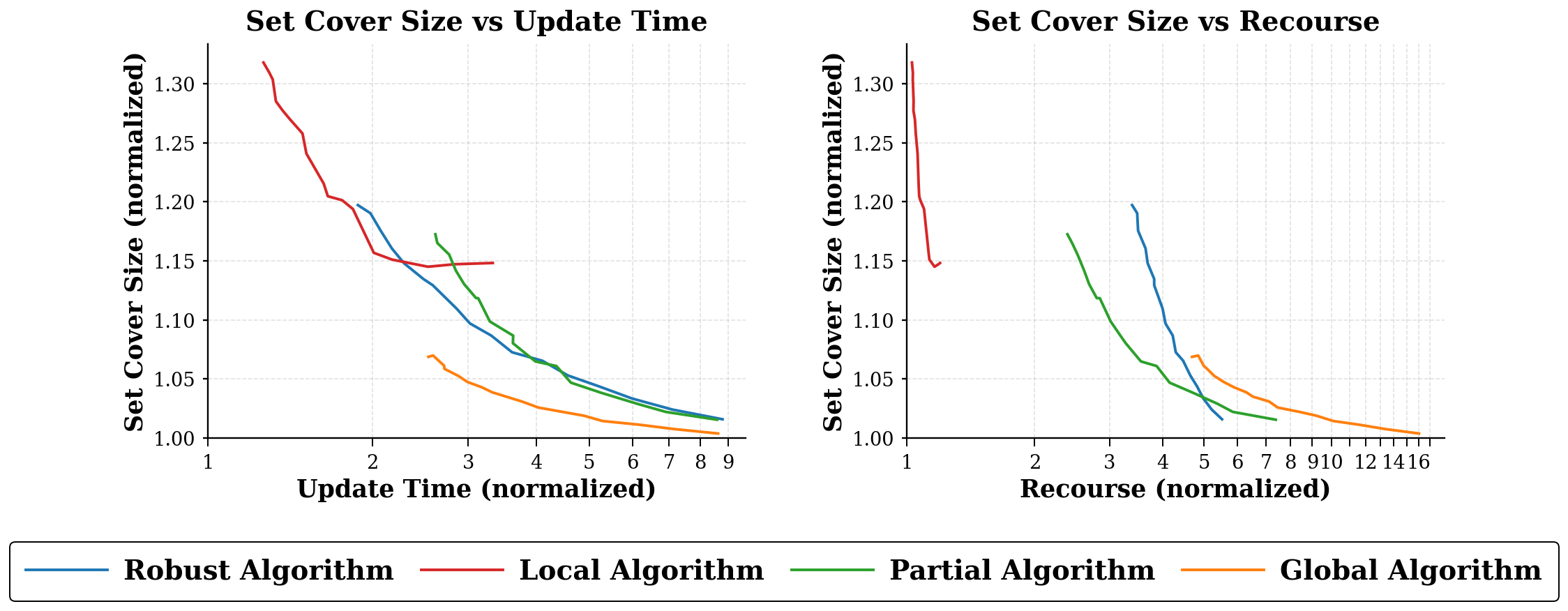}} 
    \caption{\label{fig6} Normalized trade-off curves representing solution quality vs. efficiency. The left plot displays mean normalized set cover size vs. mean normalized update time, and the right plot displays mean normalized set cover size vs. mean normalized recourse. Each curve represents the geometric mean across 120 instances for the full range of $\beta$ values; curves closer to the bottom-left corner indicate superior overall performance.}
\end{figure}

In the second step, to compare the four algorithms, we selected a representative $\beta$ for each by balancing the three quality metrics. While the optimal $\beta$ is application-dependent, we defined it as the value minimizing the objective function $g(s,t,r) = s \cdot t^{1/2} \cdot r^{1/2}$, where $s, t$, and $r$ represent amortized set cover size, update time, and recourse, respectively. We chose this function to balance solution size against the other two metrics, which are strongly positively correlated. 
For each algorithm, we identified the $\beta$ minimizing $g$ per instance and used the median value across all instances as the ``best'' $\beta$. The chosen $\beta$'s were 1.99, 1.9, 1.99 and 1.495 for the robust, local, partial and global algorithms respectively. \cref{figpp} presents the performance profiles comparing the four algorithms using these selected values.

\vspace{-2pt}

\begin{figure}[h!]
	\center{\includegraphics[scale=0.26]{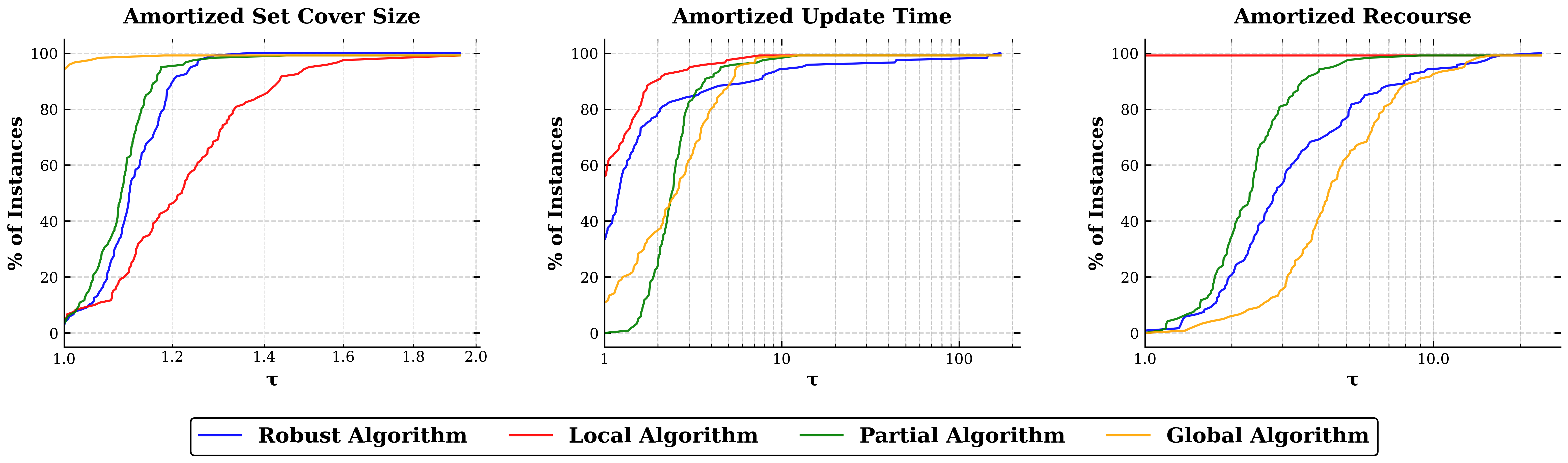}} 
    \caption{\label{figpp} Performance profiles for amortized set cover size (left), amortized update time (center) and amortized recourse (right) of all four algorithms, each with the chosen $\beta$.}
\end{figure}

\vspace{-5pt}

\noindent The global algorithm achieves the smallest average set cover in $\approx$ 95\% of instances ($\tau=1$). The partial and robust algorithms follow, while the local algorithm generally yields the lowest-quality results. For amortized recourse, the hierarchy is roughly reversed: the local algorithm always performs best (100\% at $\tau=1$), followed by the partial, robust, and global algorithms. Low recourse for the local algorithm is expected, as it maintains the solution by consistently ``patching'' it rather than triggering broader recomputations. Amortized update time reveals a more nuanced tradeoff. The local algorithm is generally the fastest, achieving the top result in $\approx$ 60\% of instances. While the partial algorithm never achieves the fastest time (0\% at $\tau=1$), it exhibits the steepest growth, indicating that it consistently remains near the best performance. Conversely, the robust algorithm is the fastest in one-third of instances but shows the least steep growth, failing to reach 100\% even at very high $\tau$ values, suggesting significant overhead in several instances. The global algorithm’s performance lies between these two. For both efficiency metrics (update time and recourse), the $\tau$ values are significantly higher than those for the quality metric (size). This indicates that solution quality is much more stable across algorithms than efficiency, mirroring the observed influence of $\beta$ on each algorithm's performance.
To summarize, generally the local algorithm offers the highest efficiency, the global algorithm the highest quality, and the robust and partial algorithms offer a balanced middle ground. All numeric results are in Tables \ref{table1}, \ref{table2}, and \ref{table3}. 

%% file: sec6.tex
\section{Conclusion} \label{concsec}

This paper provides the first empirical study of the dynamic set cover problem, bridging the gap between theory and practice. We implemented and evaluated four greedy-based algorithms: a novel and simple one, alongside three simplifications of state-of-the-art theoretical designs modified to prioritize practical efficiency over asymptotic bounds. Our evaluations focused on three amortized metrics: set cover size as a measure of quality, and update time and recourse as measures of efficiency. We assessed the tradeoff between these metrics for each algorithm, followed by a comparative analysis of the four algorithms.

\vspace{-5pt}

\subparagraph{Future Work.} 
A natural extension would be a comprehensive comparison of the four algorithms against a naive baseline that recomputes the set cover from scratch at every update. While we attempted this over several weeks by waiving the 48-hour time limit, the extreme computational overhead restricted successful runs to only the smaller instances. Preliminary results suggest that our dynamic algorithms provide drastic efficiency gains while maintaining competitive solution quality; however, full-scale benchmarking on the larger instances is required to fully verify these results.

While this study focuses on the \emph{unweighted} dynamic set cover problem for its accessible and intuitive presentation, a natural next step is the \emph{weighted} version. In this variant, each set has a weight and the goal is to maintain a minimum-weight cover. In practice, the weighted implementation would not differ significantly, as it would use the same core algorithmic procedures as a black box to maintain the solution.

A further compelling direction is the implementation of theoretical \emph{primal-dual}-based dynamic set cover algorithms designed for the low-frequency regime ($f < \ln n$). Evaluating their performance against the greedy-based algorithms presented here across all instance types (low-frequency and high-frequency) would clarify whether primal-dual approaches are indeed superior for low-frequency instances. Furthermore, such a comparison would help identify the practical transition point where greedy algorithms become the more effective choice.

%% file: secappendix.tex
\section{Appendix} \label{secapp}

\begin{figure}[h!]
	\center{\includegraphics[scale=0.26]{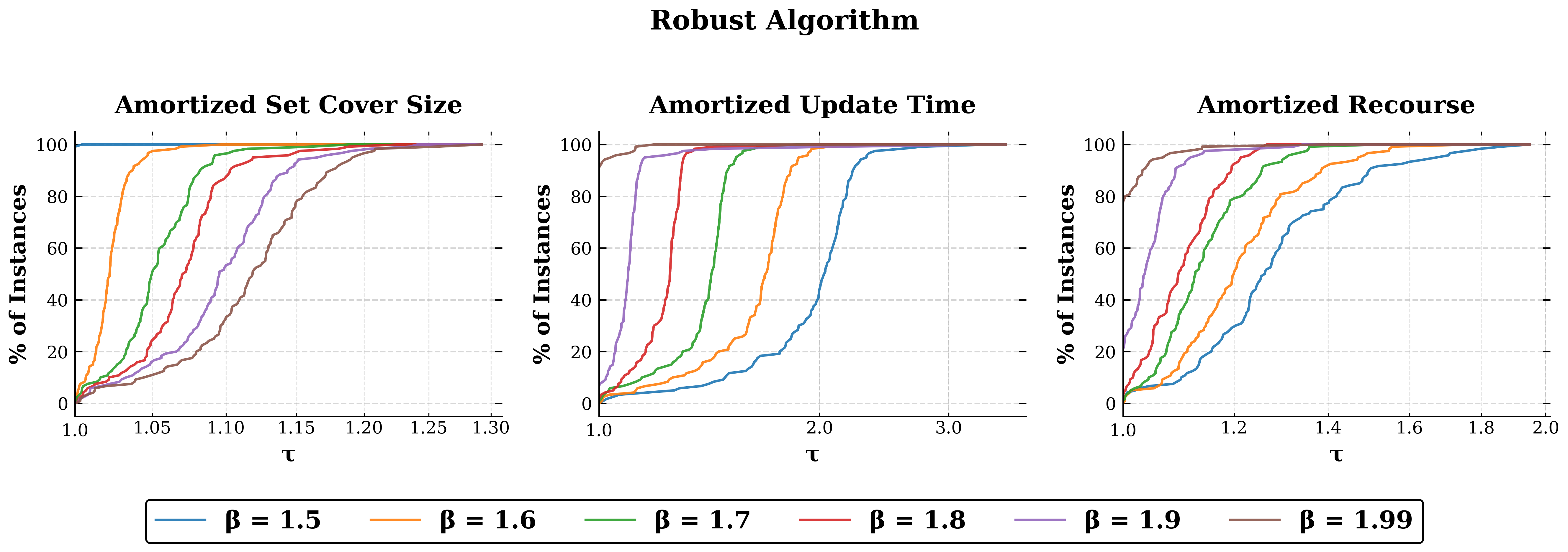}} 
\end{figure}

\vspace{-5pt}

\begin{figure}[h!]
	\center{\includegraphics[scale=0.26]{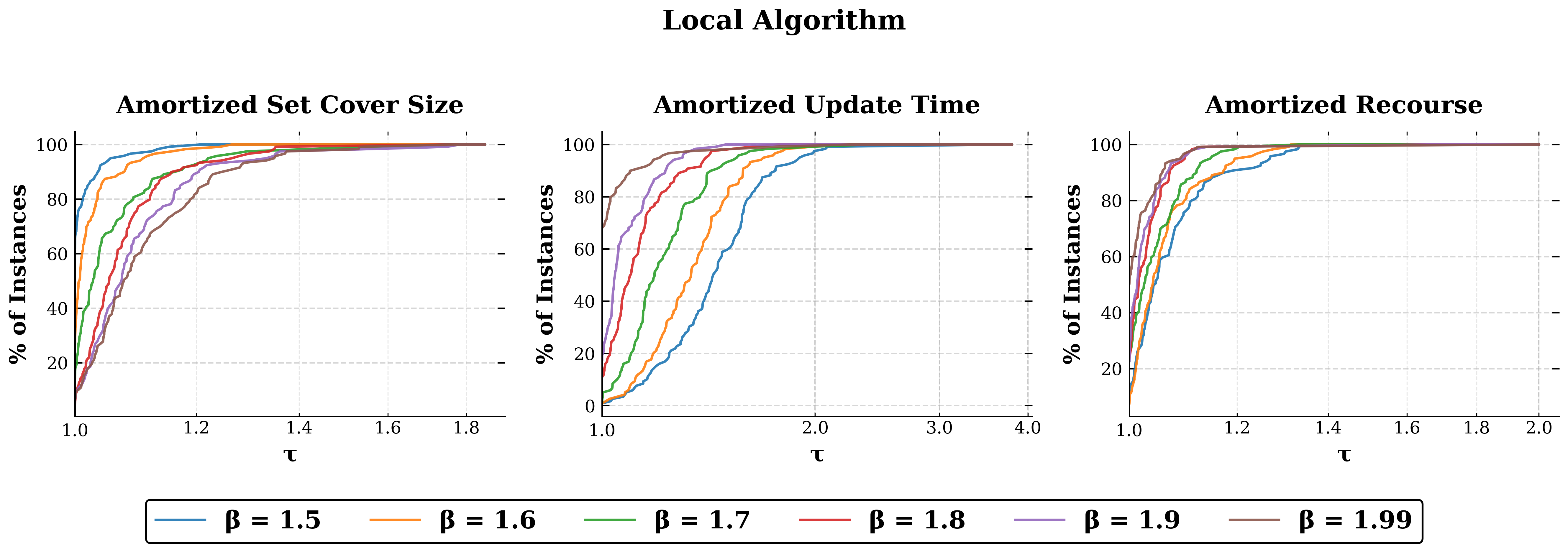}} 
\end{figure}

\vspace{-5pt}

\begin{figure}[h!]
	\center{\includegraphics[scale=0.26]{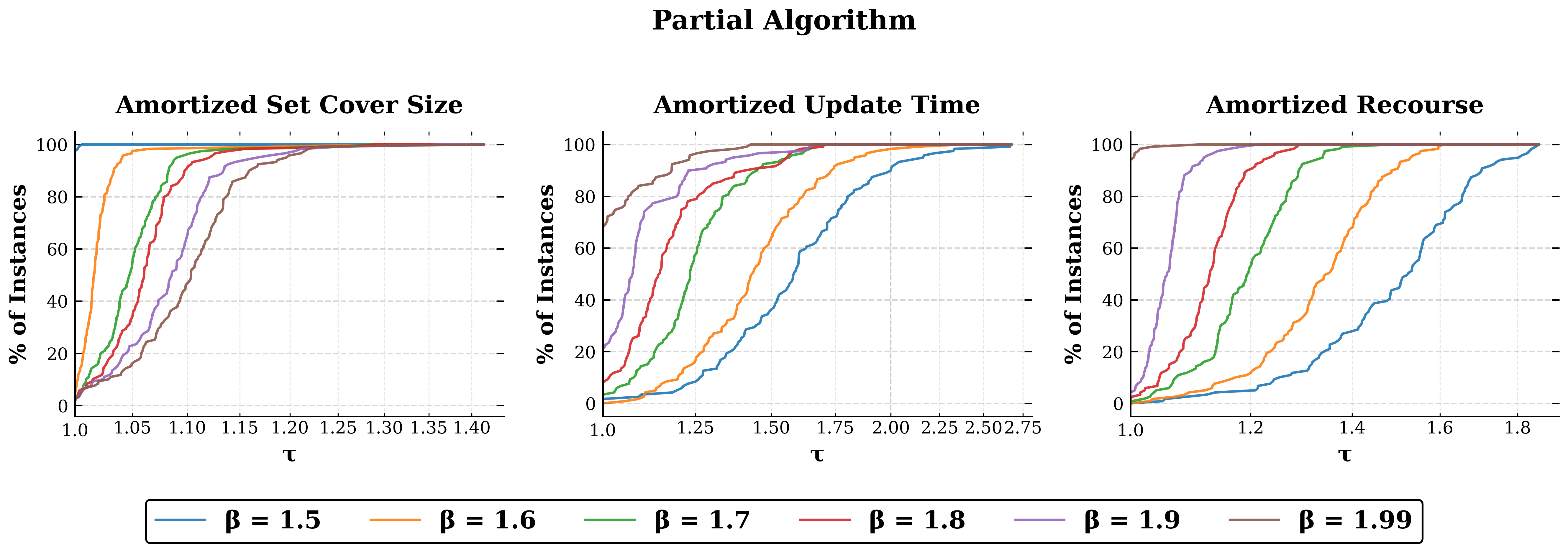}} 
\end{figure}

\vspace{-5pt}

\begin{figure}[h!]
	\center{\includegraphics[scale=0.26]{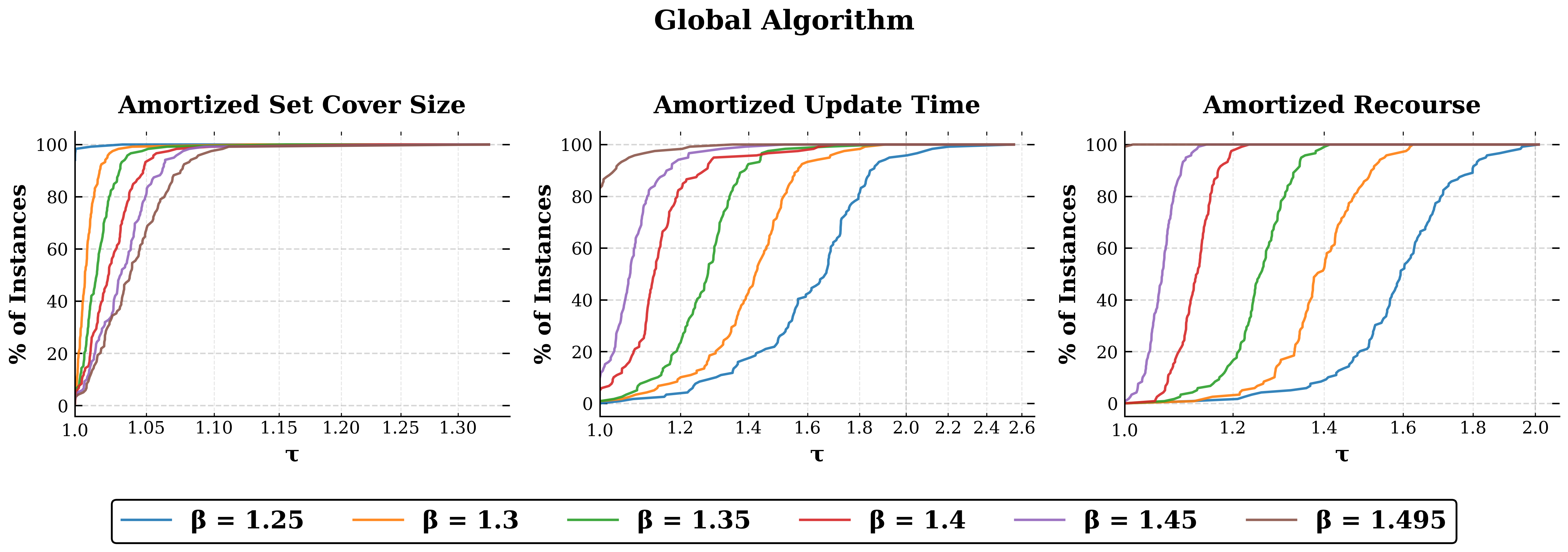}} 
    \caption{\label{figpp4} Performance profiles for each algorithm evaluated across all three quality metrics. Each row represents a specific algorithm (top to bottom: robust, local, partial, and global), divided into six variants with varying values of $\beta$. The columns represent amortized set cover size (left), update time (center), and recourse (right). Detailed legends for $\beta$ values are provided beneath each row.}
\end{figure}

\newpage

\newcolumntype{Y}{>{\centering\arraybackslash}X}

\setlength{\tabcolsep}{2.5pt}

\begin{sidewaystable}[p]
    \centering
    \footnotesize
    \caption{Statistics and results for the first 40 instances. All results are for the chosen $\beta$. $s$, $t$ and $r$ represent the amortized set cover size, update time and recourse, respectively. The subscripts rob, loc, par and glo represent the robust, local, partial and global algorithms, respectively.}
    \label{table1}
    \begin{tabularx}{\linewidth}{@{} l *{16}{Y} @{}}
        \toprule
        Ins. & $k$ & $n$ & $m$ & $f$ & $s_{\mathsf{rob}}$ & $t_{\mathsf{rob}}$(ns) & $r_{\mathsf{rob}}$ & $s_{\mathsf{loc}}$ & $t_{\mathsf{loc}}$(ns) & $r_{\mathsf{loc}}$ & $s_{\mathsf{par}}$ & $t_{\mathsf{par}}$(ns) & $r_{\mathsf{par}}$ & $s_{\mathsf{glo}}$ & $t_{\mathsf{glo}}$(ns) & $r_{\mathsf{glo}}$ \\ 
        \midrule
        1 & 5082 & 254 & 31022 & 969 & 287.753 & 5439.58 & 1.9264 & 229.997 & 5367.95 & 0.9976 & 229.973 & 23606.32 & 2.0256 & 229.989 & 17360.48 & 2.2735 \\
        2 & 6318 & 315 & 15905 & 1837 & 5.3107 & 4432988.0 & 0.7822 & 5.8957 & 295395.6 & 0.0332 & 5.1955 & 451731.4 & 0.082 & 4.4395 & 100355.76 & 0.132 \\
        3 & 9858 & 492 & 10595 & 4928 & 129.646 & 7225.73 & 1.0742 & 130.12 & 3672.21 & 0.3518 & 125.407 & 11980.24 & 1.016 & 117.572 & 14571.88 & 2.1319 \\
        4 & 12442 & 622 & 22687 & 64 & 388.074 & 4709.44 & 1.5224 & 423.81 & 3070.41 & 0.664 & 372.093 & 9170.11 & 1.3411 & 347.446 & 16667.04 & 2.3088 \\
        5 & 21162 & 1058 & 10581 & 1711 & 4.8038 & 597014.6 & 0.1294 & 5.5951 & 22575.16 & 0.0076 & 4.9958 & 27667.5 & 0.0248 & 4.5228 & 13734.0 & 0.0546 \\
        6 & 21210 & 1060 & 10605 & 147 & 493.559 & 58525.96 & 0.5902 & 402.969 & 13249.22 & 0.4395 & 402.346 & 22153.32 & 0.5204 & 400.905 & 13959.2 & 0.6058 \\
        7 & 21548 & 1077 & 10774 & 11 & 481.57 & 4013.7 & 1.4498 & 515.562 & 3626.16 & 0.5787 & 446.808 & 9093.77 & 1.4478 & 412.166 & 12453.46 & 2.5068 \\
        8 & 21696 & 1084 & 10848 & 723 & 53.3009 & 116051.8 & 0.4747 & 55.8195 & 33474.22 & 0.0745 & 54.2873 & 62388.62 & 0.2895 & 49.6107 & 38395.18 & 0.7263 \\
        9 & 22794 & 1139 & 11397 & 25 & 266.566 & 7844.93 & 1.5366 & 302.59 & 5229.61 & 0.4943 & 262.863 & 11901.06 & 0.9913 & 232.73 & 14964.96 & 1.9488 \\
        10 & 23476 & 1173 & 11738 & 10 & 492.033 & 3437.06 & 1.4334 & 527.189 & 2804.5 & 0.6025 & 462.834 & 8045.11 & 1.4809 & 422.81 & 11600.02 & 2.4291 \\
        11 & 23580 & 1179 & 11790 & 120 & 363.083 & 4192.15 & 0.9394 & 381.919 & 3616.16 & 0.4694 & 337.913 & 7705.03 & 0.9617 & 325.289 & 10764.32 & 1.5578 \\
        12 & 24010 & 1200 & 12005 & 50 & 141.853 & 9535.68 & 0.8262 & 146.363 & 7709.88 & 0.2482 & 140.207 & 13628.18 & 0.5286 & 134.122 & 11246.14 & 0.9512 \\
        13 & 24656 & 1232 & 12328 & 121 & 251.145 & 16694.0 & 1.4178 & 321.063 & 10478.06 & 0.4169 & 247.344 & 22044.62 & 1.0903 & 222.794 & 24964.76 & 1.978 \\
        14 & 24942 & 1247 & 872622 & 75355 & 81.1309 & 17736340 & 0.9063 & 92.955 & 3196086.0 & 0.1148 & 82.2677 & 4873080.0 & 0.309 & 75.4089 & 1843784.0 & 0.6323 \\
        15 & 25196 & 1259 & 12598 & 33 & 185.937 & 7710.13 & 1.2108 & 219.402 & 6284.92 & 0.4707 & 187.785 & 12027.92 & 0.9252 & 166.901 & 14449.38 & 1.9537 \\
        16 & 26628 & 1331 & 13332 & 2265 & 378.416 & 11178.3 & 0.6427 & 400.065 & 8507.19 & 0.2609 & 371.537 & 15223.34 & 0.5086 & 364.064 & 9659.76 & 0.812 \\
        17 & 26712 & 1335 & 13356 & 102 & 325.733 & 6319.12 & 1.053 & 342.499 & 3667.42 & 0.4701 & 301.547 & 9147.24 & 1.1082 & 287.222 & 10794.8 & 1.6548 \\
        18 & 27070 & 1353 & 13535 & 45 & 279.234 & 16505.0 & 1.5217 & 285.45 & 10686.1 & 0.5386 & 256.446 & 26253.96 & 1.1461 & 233.179 & 32938.62 & 2.3593 \\
        19 & 27382 & 1369 & 13682 & 16 & 38.785 & 60562.08 & 0.2975 & 42.0452 & 4346.54 & 0.0195 & 41.0081 & 10675.36 & 0.0638 & 39.4657 & 7348.85 & 0.1196 \\
        20 & 27930 & 1396 & 13965 & 81 & 53.7755 & 60389.92 & 0.9726 & 69.073 & 27803.2 & 0.227 & 56.5894 & 37932.16 & 0.352 & 49.4463 & 27013.06 & 0.7986 \\
        21 & 27984 & 1399 & 13992 & 71 & 110.164 & 23130.3 & 0.6693 & 110.897 & 15164.62 & 0.167 & 111.56 & 24851.78 & 0.6156 & 94.8106 & 19204.4 & 1.3084 \\
        22 & 28000 & 1400 & 14000 & 294 & 138.933 & 117946.0 & 1.7516 & 191.976 & 63926.84 & 0.1455 & 151.35 & 129075.0 & 0.4919 & 123.422 & 79639.46 & 1.0205 \\
        23 & 28428 & 1421 & 14214 & 23 & 218.824 & 8120.68 & 1.0419 & 240.805 & 7803.96 & 0.5428 & 209.888 & 16043.5 & 0.8477 & 197.064 & 13090.64 & 1.2028 \\
        24 & 28680 & 1434 & 14340 & 7229 & 121.762 & 661128.8 & 0.9397 & 133.866 & 720395.2 & 0.099 & 125.582 & 1190086.0 & 0.2348 & 118.329 & 96371.52 & 0.3708 \\
        25 & 29644 & 1482 & 14822 & 50 & 74.034 & 41130.12 & 0.705 & 76.7424 & 17677.26 & 0.1218 & 72.958 & 36337.02 & 0.413 & 68.5249 & 27149.96 & 1.0737 \\
        26 & 30480 & 1524 & 72600 & 20 & 143.352 & 68947.54 & 0.0974 & 141.252 & 8486.27 & 0.0736 & 141.086 & 26988.5 & 0.1085 & 141.072 & 44504.42 & 0.2437 \\
        27 & 30874 & 1543 & 19321 & 63 & 446.712 & 7778.8 & 1.5377 & 543.375 & 6081.94 & 0.5825 & 463.415 & 13707.8 & 1.1299 & 399.93 & 19673.24 & 2.2271 \\
        28 & 30898 & 1544 & 15449 & 243 & 30.6138 & 186241.6 & 0.446 & 36.5083 & 39865.64 & 0.0538 & 31.5938 & 56368.06 & 0.1527 & 28.1611 & 24032.02 & 0.3232 \\
        29 & 32092 & 1604 & 16706 & 360 & 600.745 & 7483.63 & 0.853 & 598.463 & 5605.71 & 0.3616 & 570.373 & 13116.2 & 0.8328 & 546.345 & 16094.64 & 1.4241 \\
        30 & 32292 & 1614 & 16146 & 89 & 62.9686 & 53537.52 & 0.6633 & 71.9663 & 20541.48 & 0.1279 & 64.2677 & 37139.38 & 0.3511 & 58.7613 & 30429.48 & 0.9428 \\
        31 & 33566 & 1678 & 16783 & 14672 & 4.5682 & 6224774.0 & 0.1951 & 3.3493 & 208021.8 & 0.0162 & 4.8707 & 421492.8 & 0.0243 & 3.9648 & 42573.1 & 0.0275 \\
        32 & 34444 & 1722 & 17222 & 19 & 305.2 & 7201.02 & 0.9005 & 401.459 & 5170.08 & 0.4385 & 316.504 & 9731.86 & 0.8177 & 283.107 & 11925.94 & 1.4998 \\
        33 & 34722 & 1736 & 17361 & 360 & 54.6294 & 72249.52 & 0.4594 & 54.4614 & 19281.2 & 0.0552 & 54.5105 & 32236.98 & 0.268 & 49.874 & 22617.64 & 0.7517 \\
        34 & 40000 & 2000 & 20000 & 29 & 232.307 & 28545.6 & 2.0232 & 300.394 & 17879.46 & 0.4068 & 232.741 & 27688.12 & 0.7719 & 200.726 & 37413.16 & 1.7052 \\
        35 & 41792 & 2089 & 20896 & 12 & 996.075 & 6902.75 & 1.8538 & 1148.0 & 7179.37 & 0.6057 & 999.954 & 15643.24 & 1.4432 & 893.392 & 21968.76 & 2.8161 \\
        36 & 42072 & 2103 & 21036 & 28 & 800.972 & 10293.52 & 1.6462 & 895.052 & 8660.18 & 0.5157 & 794.798 & 19269.38 & 1.2963 & 697.87 & 27277.34 & 2.3827 \\
        37 & 45120 & 2256 & 22560 & 112 & 123.345 & 32152.92 & 0.4948 & 128.751 & 23294.42 & 0.1341 & 120.369 & 32639.78 & 0.4695 & 113.73 & 22213.24 & 0.8515 \\
        38 & 45390 & 2269 & 22695 & 108 & 131.388 & 49422.1 & 1.5047 & 168.574 & 31148.32 & 0.3229 & 137.817 & 47887.34 & 0.6042 & 117.817 & 36055.26 & 1.2281 \\
        39 & 46824 & 2341 & 23412 & 15 & 499.27 & 6499.2 & 1.4433 & 502.693 & 6838.26 & 0.7461 & 461.733 & 12886.74 & 1.2276 & 415.601 & 17960.76 & 2.2906 \\
        40 & 47120 & 2356 & 23560 & 21 & 268.884 & 13449.04 & 1.4885 & 322.396 & 11399.88 & 0.534 & 251.649 & 20568.96 & 0.938 & 228.249 & 24901.46 & 1.8933 \\
        \bottomrule
    \end{tabularx}
\end{sidewaystable}

\begin{sidewaystable}[p]
    \centering
    \footnotesize
    \caption{Statistics and results for the second 40 instances. All results are for the chosen $\beta$. $s$, $t$ and $r$ represent the amortized set cover size, update time and recourse, respectively. The subscripts rob, loc, par and glo represent the robust, local, partial and global algorithms, respectively.}
    \label{table2}
    \begin{tabularx}{\linewidth}{@{} l *{16}{Y} @{}}
        \toprule
        Ins. & $k$ & $n$ & $m$ & $f$ & $s_{\mathsf{rob}}$ & $t_{\mathsf{rob}}$(ns) & $r_{\mathsf{rob}}$ & $s_{\mathsf{loc}}$ & $t_{\mathsf{loc}}$(ns) & $r_{\mathsf{loc}}$ & $s_{\mathsf{par}}$ & $t_{\mathsf{par}}$(ns) & $r_{\mathsf{par}}$ & $s_{\mathsf{glo}}$ & $t_{\mathsf{glo}}$(ns) & $r_{\mathsf{glo}}$ \\ 
        \midrule
        41 & 49082 & 2454 & 27240 & 562 & 625.902 & 9307.64 & 0.6807 & 660.548 & 6924.21 & 0.2466 & 617.088 & 11296.98 & 0.5046 & 595.109 & 11585.72 & 0.8174 \\
        42 & 49684 & 2484 & 24842 & 181 & 815.195 & 3815.55 & 1.0488 & 769.687 & 3039.94 & 0.4521 & 746.872 & 8502.3 & 1.2439 & 711.249 & 14896.42 & 2.4125 \\
        43 & 50374 & 2518 & 25187 & 22767 & 174.51 & 43786.98 & 0.2217 & 206.162 & 3400.64 & 0.0931 & 180.249 & 7298.35 & 0.1745 & 174.105 & 8408.98 & 0.2884 \\
        44 & 55386 & 2769 & 67593 & 15 & 1399.49 & 5260.27 & 1.9145 & 2337.1 & 3612.38 & 0.8711 & 1412.34 & 10485.92 & 1.5539 & 1200.35 & 21386.36 & 2.6747 \\
        45 & 58134 & 2906 & 29067 & 345 & 314.489 & 68813.42 & 1.1207 & 335.838 & 34598.06 & 0.0785 & 318.418 & 84512.86 & 0.4671 & 290.921 & 41045.6 & 0.8276 \\
        46 & 58564 & 2928 & 29282 & 18 & 686.571 & 7359.38 & 1.1924 & 640.005 & 6167.8 & 0.6086 & 635.622 & 15335.74 & 1.252 & 583.915 & 21675.04 & 2.3749 \\
        47 & 59840 & 2992 & 29920 & 181 & 319.474 & 156509.8 & 0.4433 & 294.381 & 26572.56 & 0.1382 & 290.356 & 58197.24 & 0.3901 & 288.642 & 40787.36 & 0.8214 \\
        48 & 59914 & 2995 & 29957 & 154 & 452.936 & 45564.02 & 1.3872 & 481.115 & 28505.94 & 0.3979 & 428.234 & 63750.36 & 0.9428 & 386.377 & 62983.1 & 1.9834 \\
        49 & 59990 & 2999 & 29995 & 8 & 719.544 & 3548.32 & 0.9202 & 744.789 & 3311.07 & 0.678 & 647.689 & 8380.51 & 1.306 & 573.906 & 14355.38 & 2.4728 \\
        50 & 60324 & 3016 & 27944 & 120 & 372.29 & 20323.88 & 1.4106 & 473.213 & 14440.36 & 0.2528 & 382.513 & 27533.26 & 0.5279 & 344.009 & 19768.64 & 1.1143 \\
        51 & 61596 & 3079 & 30798 & 15054 & 1085.05 & 2536456.0 & 0.9414 & 1465.63 & 20733.58 & 0.4965 & 1074.85 & 34748.68 & 0.8537 & 925.56 & 14836.0 & 1.497 \\
        52 & 65538 & 3276 & 32769 & 73 & 611.819 & 8467.07 & 0.983 & 649.748 & 5053.44 & 0.2886 & 623.542 & 12727.86 & 0.7041 & 553.48 & 19463.96 & 1.3422 \\
        53 & 69840 & 3492 & 34920 & 438 & 93.2377 & 177874.0 & 0.4484 & 98.4801 & 52920.5 & 0.0541 & 96.4026 & 105585.8 & 0.2169 & 87.2432 & 55472.94 & 0.6411 \\
        54 & 72000 & 3600 & 36000 & 519 & 83.1629 & 677449.4 & 0.7467 & 98.2479 & 190210.8 & 0.0805 & 83.2803 & 494395.6 & 0.2549 & 74.4111 & 167191.8 & 0.6039 \\
        55 & 72834 & 3641 & 36417 & 204 & 92.1224 & 143161.6 & 0.5019 & 94.7967 & 49827.3 & 0.098 & 93.3368 & 102867.6 & 0.3338 & 84.0223 & 52800.12 & 0.8015 \\
        56 & 72882 & 3644 & 36441 & 21 & 430.192 & 9133.73 & 1.3889 & 550.951 & 7180.16 & 0.3853 & 427.316 & 16499.78 & 0.9336 & 386.526 & 19152.56 & 1.5548 \\
        57 & 75524 & 3776 & 37762 & 42 & 230.075 & 21419.74 & 1.0404 & 228.548 & 11830.38 & 0.2012 & 228.978 & 27324.56 & 0.5877 & 210.938 & 24545.64 & 1.2327 \\
        58 & 76830 & 3841 & 38415 & 25 & 370.753 & 14288.6 & 1.353 & 384.167 & 12169.92 & 0.8967 & 372.303 & 24290.32 & 1.1529 & 331.347 & 26618.62 & 1.6782 \\
        59 & 77488 & 3874 & 38744 & 126 & 338.292 & 24759.72 & 1.2871 & 351.147 & 20005.28 & 0.3723 & 312.861 & 35750.32 & 0.7343 & 285.503 & 31002.66 & 1.6101 \\
        60 & 80432 & 4021 & 40216 & 20024 & 1416.94 & 2139940.0 & 0.7985 & 1936.55 & 14913.24 & 0.499 & 1401.5 & 24006.06 & 0.8433 & 1210.61 & 22775.22 & 1.5658 \\
        61 & 91816 & 4590 & 59498 & 155 & 2119.28 & 11169.46 & 1.2688 & 1997.2 & 14480.72 & 0.6008 & 1832.86 & 28294.0 & 1.3832 & 1790.24 & 38953.84 & 1.9314 \\
        62 & 93670 & 4683 & 46835 & 145 & 280.704 & 41693.28 & 0.6525 & 293.8 & 34259.32 & 0.1286 & 275.506 & 56622.5 & 0.6023 & 253.137 & 52783.56 & 1.6431 \\
        63 & 98304 & 4915 & 49152 & 39 & 322.384 & 36298.34 & 1.0097 & 382.34 & 30151.68 & 0.2108 & 335.429 & 48212.82 & 0.5105 & 292.848 & 52119.24 & 1.3681 \\
        64 & 99332 & 4966 & 49666 & 13 & 1411.68 & 11181.86 & 1.6046 & 1356.6 & 12454.44 & 0.9019 & 1205.61 & 24501.94 & 1.5696 & 1127.74 & 36975.72 & 2.6239 \\
        65 & 100000 & 5000 & 50000 & 56 & 1429.04 & 26476.1 & 1.8846 & 1765.89 & 31928.68 & 0.3916 & 1532.19 & 54616.56 & 1.0679 & 1254.71 & 64572.92 & 2.3124 \\
        66 & 103986 & 5199 & 51993 & 23 & 2244.17 & 5764.94 & 1.2888 & 2092.49 & 4778.71 & 0.7251 & 1936.21 & 13221.28 & 1.23 & 1857.95 & 24396.42 & 2.0975 \\
        67 & 108038 & 5401 & 54019 & 1426 & 608.248 & 10551.14 & 0.8116 & 702.365 & 8952.75 & 0.2926 & 612.864 & 17900.26 & 0.7101 & 551.865 & 28306.18 & 1.8233 \\
        68 & 109740 & 5487 & 54870 & 276 & 246.101 & 40571.08 & 0.7167 & 313.446 & 19827.44 & 0.2365 & 255.068 & 58393.98 & 0.5453 & 218.236 & 39506.86 & 1.2134 \\
        69 & 110642 & 5532 & 65536 & 17998 & 821.432 & 59605.04 & 0.2977 & 851.236 & 58062.8 & 0.1307 & 817.135 & 149445.6 & 0.2035 & 808.524 & 28235.44 & 0.278 \\
        70 & 120000 & 6000 & 60000 & 36 & 2507.76 & 10453.78 & 1.7553 & 2940.96 & 11179.58 & 0.5699 & 2549.41 & 24881.2 & 1.3139 & 2208.49 & 40767.56 & 2.5131 \\
        71 & 121480 & 6074 & 60740 & 63 & 476.257 & 21366.28 & 1.0353 & 583.219 & 20764.12 & 0.3997 & 502.888 & 37234.38 & 0.6978 & 441.649 & 48738.72 & 1.678 \\
        72 & 124848 & 6242 & 62424 & 44 & 462.051 & 17984.1 & 0.6377 & 440.016 & 12976.92 & 0.2227 & 459.695 & 35121.7 & 0.5946 & 424.831 & 34267.44 & 1.4813 \\
        73 & 130050 & 6502 & 65025 & 25 & 1680.17 & 11666.68 & 1.0524 & 1638.76 & 8729.21 & 0.4631 & 1576.03 & 30429.82 & 1.282 & 1459.01 & 40958.36 & 2.2551 \\
        74 & 131072 & 6553 & 147266 & 9 & 2403.97 & 6650.09 & 1.4349 & 2600.77 & 6507.06 & 0.5834 & 2436.47 & 18177.28 & 1.3217 & 1958.68 & 46429.1 & 2.8214 \\
        75 & 132254 & 6612 & 66127 & 27 & 649.921 & 16117.12 & 1.5108 & 734.344 & 14118.62 & 0.7324 & 645.14 & 29710.46 & 0.9834 & 571.841 & 37382.3 & 1.9701 \\
        76 & 134048 & 6702 & 67024 & 37 & 695.539 & 29899.42 & 1.8317 & 866.885 & 25726.96 & 0.3843 & 687.702 & 41471.36 & 0.7647 & 585.256 & 55100.22 & 1.7964 \\
        77 & 134346 & 6717 & 67173 & 89 & 1825.25 & 20281.18 & 0.7654 & 1763.83 & 39187.32 & 0.556 & 1669.66 & 64928.06 & 0.6576 & 1627.93 & 32812.94 & 0.8759 \\
        78 & 136242 & 6812 & 68121 & 68121 & 641.157 & 58528.46 & 1.428 & 712.806 & 53892.58 & 0.8876 & 644.97 & 80256.3 & 1.0534 & 569.416 & 52729.84 & 1.7393 \\
        79 & 138470 & 6923 & 132402 & 12 & 2210.99 & 10213.44 & 1.8489 & 2293.49 & 9309.22 & 0.9065 & 1922.55 & 24556.08 & 1.5748 & 1716.88 & 50530.36 & 2.9107 \\
        80 & 141312 & 7065 & 70656 & 33 & 690.723 & 17319.5 & 1.4105 & 793.481 & 15555.84 & 0.5777 & 689.886 & 35562.68 & 0.8948 & 607.741 & 40741.6 & 1.8693 \\

        \bottomrule
    \end{tabularx}
\end{sidewaystable}

\begin{sidewaystable}[p]
    \centering
    \footnotesize
    \caption{Statistics and results for the last 40 instances. All results are for the chosen $\beta$. $s$, $t$ and $r$ represent the amortized set cover size, update time and recourse, respectively. The subscripts rob, loc, par and glo represent the robust, local, partial and global algorithms, respectively.}
    \label{table3}
    \begin{tabularx}{\linewidth}{@{} l *{16}{Y} @{}}
        \toprule
        Ins. & $k$ & $n$ & $m$ & $f$ & $s_{\mathsf{rob}}$ & $t_{\mathsf{rob}}$(ns) & $r_{\mathsf{rob}}$ & $s_{\mathsf{loc}}$ & $t_{\mathsf{loc}}$(ns) & $r_{\mathsf{loc}}$ & $s_{\mathsf{par}}$ & $t_{\mathsf{par}}$(ns) & $r_{\mathsf{par}}$ & $s_{\mathsf{glo}}$ & $t_{\mathsf{glo}}$(ns) & $r_{\mathsf{glo}}$ \\ 
        \midrule
        81 & 143010 & 7150 & 71505 & 345 & 406.119 & 86216.38 & 0.5872 & 422.233 & 41982.32 & 0.085 & 413.092 & 101455.2 & 0.428 & 365.102 & 75259.84 & 1.2086 \\
        82 & 148124 & 7406 & 74062 & 565 & 1159.82 & 13134.08 & 1.0639 & 1449.92 & 15502.24 & 0.3562 & 1174.96 & 32477.86 & 0.8353 & 1047.17 & 33857.38 & 1.7232 \\
        83 & 148208 & 7410 & 74104 & 271 & 3087.19 & 4854.94 & 1.5042 & 3737.19 & 4993.87 & 0.8528 & 2760.14 & 13686.02 & 1.4677 & 2597.86 & 27076.82 & 2.3288 \\
        84 & 149504 & 7475 & 74752 & 55 & 1369.94 & 8009.24 & 0.7135 & 1507.18 & 5451.76 & 0.3946 & 1256.91 & 15235.08 & 0.8653 & 1224.47 & 28892.46 & 1.7676 \\
        85 & 165308 & 8265 & 82654 & 11 & 2232.74 & 12981.92 & 1.0016 & 2128.83 & 20326.32 & 0.4931 & 2082.32 & 33490.1 & 1.2002 & 1906.97 & 49537.1 & 2.11 \\
        86 & 171246 & 8562 & 85623 & 145 & 1297.64 & 40428.1 & 1.0856 & 1447.73 & 60738.0 & 0.1638 & 1317.96 & 114105.2 & 0.5942 & 1187.86 & 75020.18 & 1.1062 \\
        87 & 175608 & 8780 & 87804 & 132 & 392.276 & 58128.84 & 0.5105 & 405.038 & 73533.52 & 0.078 & 391.8 & 100322.66 & 0.3123 & 372.918 & 69685.8 & 1.0038 \\
        88 & 178800 & 8940 & 89400 & 13 & 1513.22 & 17057.12 & 1.429 & 1589.41 & 30912.8 & 0.6426 & 1402.72 & 47089.74 & 1.1074 & 1226.39 & 59792.66 & 1.9877 \\
        89 & 180898 & 9044 & 90449 & 42 & 422.812 & 40805.1 & 0.5694 & 453.53 & 49843.22 & 0.1519 & 429.618 & 90117.7 & 0.3927 & 399.007 & 70889.02 & 1.1179 \\
        90 & 202984 & 10149 & 101492 & 31 & 1421.63 & 24747.46 & 1.0418 & 1621.42 & 40231.24 & 0.3821 & 1372.5 & 66061.54 & 0.9159 & 1248.71 & 76300.0 & 1.9331 \\
        91 & 204316 & 10215 & 102158 & 10 & 4840.37 & 18862.28 & 1.5964 & 5327.28 & 32273.54 & 0.5937 & 4686.8 & 53719.0 & 1.4096 & 4230.74 & 71840.14 & 2.5552 \\
        92 & 226152 & 11307 & 113076 & 7031 & 2006.37 & 27508.56 & 1.2961 & 2040.4 & 35498.72 & 0.4657 & 1868.42 & 98471.7 & 1.1365 & 1689.57 & 79854.74 & 2.3881 \\
        93 & 233670 & 11683 & 116835 & 114190 & 224.66 & 59180.78 & 0.0748 & 223.176 & 7584.04 & 0.0228 & 223.176 & 16260.32 & 0.0492 & 223.176 & 32563.5 & 0.0578 \\
        94 & 306452 & 15322 & 153226 & 5776 & 2555.71 & 30755.04 & 1.2261 & 2643.92 & 45486.6 & 0.4635 & 2375.3 & 126643.2 & 1.1055 & 2157.53 & 109772.8 & 2.484 \\
        95 & 307492 & 15374 & 153746 & 99 & 533.182 & 60106.94 & 0.6453 & 608.962 & 51563.38 & 0.1476 & 535.917 & 136503.0 & 0.4574 & 478.42 & 75068.42 & 1.1144 \\
        96 & 314928 & 15746 & 157464 & 27 & 1383.81 & 27096.6 & 0.5701 & 1455.32 & 38598.02 & 0.1872 & 1333.02 & 74364.4 & 0.6482 & 1252.87 & 78422.9 & 1.6681 \\
        97 & 320000 & 16000 & 160000 & 11 & 3620.89 & 29345.86 & 1.4398 & 3595.41 & 45234.46 & 0.8537 & 3141.74 & 76200.16 & 1.4134 & 2850.87 & 101504.0 & 2.5087 \\
        98 & 338820 & 16941 & 169410 & 719 & 1702.34 & 64134.36 & 0.9939 & 1943.63 & 100789.38 & 0.196 & 1715.91 & 153298.0 & 0.6151 & 1526.45 & 136565.8 & 1.2634 \\
        99 & 365460 & 18273 & 182730 & 11 & 5950.63 & 43324.82 & 1.1927 & 7495.71 & 71717.64 & 0.6596 & 5777.74 & 97580.34 & 1.2726 & 5032.84 & 122781.8 & 2.5065 \\
        100 & 429530 & 21476 & 214765 & 40 & 4908.73 & 20879.36 & 1.348 & 5297.32 & 29628.34 & 0.4594 & 4771.41 & 55856.74 & 1.0863 & 4178.31 & 93535.2 & 2.1648 \\
        101 & 491748 & 24587 & 245874 & 54 & 846.656 & 53518.04 & 0.4417 & 1431.22 & 42838.78 & 0.1703 & 949.071 & 104951.8 & 0.3875 & 812.423 & 109425.8 & 1.0556 \\
        102 & 518312 & 25915 & 259156 & 22 & 3982.51 & 18030.02 & 1.4096 & 4228.33 & 25088.36 & 0.5715 & 3520.92 & 74235.9 & 1.0992 & 3216.36 & 98662.32 & 2.1621 \\
        103 & 524282 & 26214 & 262144 & 31 & 4703.25 & 12694.24 & 0.8807 & 4564.69 & 20417.16 & 0.5115 & 4324.13 & 57335.86 & 1.0904 & 3908.63 & 90050.3 & 2.1197 \\
        104 & 584016 & 29200 & 292008 & 60 & 2500.24 & 29418.26 & 1.1636 & 2655.77 & 38905.58 & 0.2292 & 2510.88 & 72323.46 & 0.5646 & 2194.32 & 101882.8 & 1.3423 \\
        105 & 607234 & 30361 & 896308 & 1334 & 10866.3 & 34005.82 & 0.9041 & 11630.8 & 32709.22 & 0.3754 & 10554.3 & 62779.84 & 0.7798 & 10200.4 & 177782.2 & 1.3666 \\
        106 & 654124 & 32706 & 37830 & 19 & 29.4237 & 923016.8 & 0.0079 & 28.0354 & 104180.6 & 0.0005 & 29.3257 & 138524.2 & 0.0045 & 29.27 & 101750.0 & 0.0065 \\
        107 & 698886 & 34944 & 822922 & 403 & 11837.0 & 138365.4 & 2.0057 & 14988.6 & 200589.2 & 0.7014 & 12752.4 & 363747.2 & 1.244 & 9928.8 & 500255.2 & 2.7267 \\
        108 & 763378 & 38168 & 381689 & 295 & 11514.6 & 160202.6 & 0.4742 & 10628.2 & 351920.6 & 0.3616 & 10577.1 & 394392.6 & 0.4078 & 10662.6 & 306052.4 & 0.5432 \\
        109 & 831726 & 41586 & 415863 & 77 & 2356.02 & 163562.0 & 0.4565 & 2519.89 & 343490.8 & 0.0811 & 2376.92 & 385546.4 & 0.3614 & 2205.85 & 349216.2 & 1.2871 \\
        110 & 1007250 & 50362 & 503625 & 35 & 2807.68 & 38505.8 & 0.365 & 2363.44 & 62161.28 & 0.1405 & 2686.76 & 126599.2 & 0.4395 & 2506.75 & 178872.6 & 1.4029 \\
        111 & 1214464 & 60723 & 607232 & 158 & 3883.13 & 68688.88 & 0.7889 & 4171.96 & 93602.56 & 0.1907 & 3829.85 & 213712.2 & 0.5508 & 3494.76 & 255071.2 & 1.4647 \\
        112 & 1318066 & 65903 & 659033 & 628 & 7089.42 & 78074.06 & 0.2974 & 7198.15 & 283577.6 & 0.1574 & 6820.6 & 276626.8 & 0.3874 & 6430.3 & 390284.8 & 1.0178 \\
        113 & 1846272 & 92313 & 923136 & 57 & 4814.83 & 91937.32 & 0.8403 & 5055.76 & 117038.4 & 0.1857 & 4927.83 & 337349.6 & 0.4967 & 4478.67 & 381443.8 & 1.3391 \\
        114 & 1980004 & 99000 & 41270 & 2497 & 3136.39 & 46928.92 & 0.0374 & 3184.53 & 81720.74 & 0.028 & 3141.58 & 129513.2 & 0.0335 & 3133.94 & 123934.0 & 0.0402 \\
        115 & 2279810 & 113990 & 1139905 & 11468 & 11760.4 & 204946.0 & 0.572 & 12206.8 & 246626.6 & 0.156 & 11467.0 & 372222.6 & 0.4568 & 10944.6 & 447341.4 & 1.0527 \\
        116 & 2930274 & 146513 & 1465137 & 189 & 23801.2 & 872151.0 & 1.251 & 27701.6 & 2266110.0 & 0.5412 & 23256.5 & 3230640.0 & 1.0101 & 20491.3 & 3751330.0 & 2.3112 \\
        117 & 6692166 & 334608 & 49677 & 145 & 6802.45 & 2500250.0 & 0.1065 & 7221.99 & 7543340.0 & 0.0135 & 6932.33 & 6742320.0 & 0.0245 & 6833.71 & 3588050.0 & 0.0351 \\
        118 & 7084800 & 354240 & 3542400 & 28 & 32240.8 & 199516.0 & 1.479 & 31485.8 & 955331.0 & 0.808 & 29769.9 & 1562670.0 & 1.1292 & 27116.9 & 2427550.0 & 1.9949 \\
        119 & 8570726 & 428536 & 2268264 & 9350 & 59481.3 & 221688.0 & 0.8209 & 70195.9 & 1346220.0 & 0.1435 & 62150.6 & 1610040.0 & 0.3757 & 58963.1 & 1563890.0 & 0.5907 \\
        120 & 13840612 & 692030 & 9845725 & 3841 & 109517.0 & 5769920.0 & 0.2968 & OOT & OOT & OOT & OOT & OOT & OOT & OOT & OOT & OOT \\
        \bottomrule
    \end{tabularx}
\end{sidewaystable}